\documentclass[11pt,a4paper]{scrartcl}

\usepackage{amsmath}
\usepackage{afterpage}
\usepackage{amsfonts}
\usepackage{adjustbox}
\usepackage[a4paper]{geometry}
\usepackage{subfigure}
\usepackage{color} 
\usepackage{xr}
\usepackage{hyperref}
\usepackage{natbib}

\usepackage{etoolbox}
\makeatletter
\patchcmd{\@makecaption}
  {\parbox}
  {\advance\@tempdima-\fontdimen2} 
  {}{}
\makeatother  

\newtheorem{theorem}{Theorem}[section]
\newtheorem{proposition}[theorem]{Proposition}
\newtheorem{remark}{Remark}[section]
\newtheorem{ass}{{Assumption $\mathcal{A}$.\negthickspace}}
\newtheorem{lemma}[theorem]{Lemma}

\newcommand{\e}{\exp}
\newcommand{\var}{\operatorname{Var}}
\newcommand{\E}{\operatorname{E}}
\newcommand{\cov}{\operatorname{Cov}}
\newcommand{\ls}{\leq}
\newcommand{\gs}{\geq}
\newcommand{\eps}{\epsilon}

\newcommand{\bY}{\mathbf{Y}}
\newcommand{\bZ}{\mathbf{Z}}

\newcommand{\setZ}{\mathbb{Z}}
\newcommand{\Z}{\bZ}
\newcommand{\xdownarrow}[1]{{\left\downarrow\vbox to #1{}\right.\kern-\nulldelimiterspace}}
\newcommand{\btheta}{\boldsymbol\theta}
\newcommand{\brho}{\boldsymbol\rho}
\newcommand{\ovunset}[3]{\overset{#2}{\underset{#3}{#1}}}
\newcommand{\deq}{\overset{\mathcal{D}}{=}}

\newenvironment{proof}[1][\negthickspace]{\textbf{Proof\thickspace#1.} }{\
\rule{0.5em}{0.5em}}

\newcommand{\cb}{ }

\begin{document}

\title{Beyond Whittle: Nonparametric correction of a
  parametric likelihood with a focus on Bayesian time series analysis}


\author{Claudia Kirch\footnote{Otto-von-Guericke University,
Magdeburg,
Germany, \href{mailto:claudia.kirch@ovgu.de}{claudia.kirch@ovgu.de}}, 
Matthew C. Edwards\footnote{The University of Auckland,
  Auckland,
  New Zealand and
  Carleton College, 
Northfield, 
Minnesota, 
USA, \href{mailto:matt.edwards@auckland.ac.nz}{matt.edwards@auckland.ac.nz}},\\ 
Alexander Meier\footnote{Otto-von-Guericke University,
Magdeburg,
Germany, \href{mailto:alexander.meier@ovgu.de}{alexander.meier@ovgu.de}} and 
Renate Meyer\footnote{The University of Auckland,
Auckland,
New Zealand, \href{mailto:meyer@stat.auckland.ac.nz}{meyer@stat.auckland.ac.nz}}}



\maketitle

\begin{abstract}
  \textbf{Abstract.}
  The Whittle likelihood is widely used for Bayesian nonparametric
  estimation of the spectral density of stationary time
  series. However, the loss of efficiency for non-Gaussian time series
  can be substantial.  On the other hand, parametric methods are more
  powerful if the model is well-specified, but may fail entirely
  otherwise.  Therefore, we suggest a nonparametric correction of a
  parametric likelihood taking advantage of the efficiency of
  parametric models while mitigating sensitivities through a
  nonparametric amendment.  Using a Bernstein-Dirichlet prior for the
  nonparametric spectral correction, we show posterior consistency and
  illustrate the performance of our procedure in a simulation study
  and with LIGO gravitational wave data.
\end{abstract}

  \section{Introduction}
Statistical models can be broadly classified into {\em parametric} and
{\em nonparametric} models.  Parametric models, indexed by a finite
dimensional set of parameters, are focused, easy to analyse and have
the big advantage that when correctly specified, they will be very
efficient and powerful.  However, they can be sensitive to
misspecifications and even mild deviations of the data from the
assumed parametric model can lead to unreliabilities of inference
procedures.  Nonparametric models, on the other hand, do not rely on
data belonging to any particular family of distributions. As they make
fewer assumptions, their applicability is much wider than that of
corresponding parametric methods. However, this generally comes at the
cost of reduced efficiency compared to parametric models.

Standard time series literature is dominated by parametric models like
autoregressive integrated moving average models \citep{box2013time},
the more recent autoregressive conditional heteroskedasticity models
for time-varying volatility \citep{Engle82,Bollerslev86}, state-space
\citep{Durbin01}, and Markov switching models \citep{Bauwens99}.  In
particular, Bayesian time series analysis \citep{Steel08} is
inherently parametric in that a completely specified likelihood
function is needed.  Nonetheless, the use of nonparametric techniques
has a long tradition in time series analysis. \citep{Schuster1898}
introduced the periodogram which may be regarded as the origin of
spectral analysis and a classical nonparametric tool for time series
analysis \citep{Haerdle97}. {\em Frequentist} time series analyses
especially use nonparametric methods \citep{FanYao03, Wasserman06}
including a variety of bootstrap methods, computer-intensive
resampling techniques initially introduced for independent data, that
have been taylored to and specifically developed for time series
\citep{Haerdle03, Kreiss11,Kreiss11a}. An important class of
nonparametric methods is based on frequency domain techniques, most
prominently smoothing the periodogram. These include a variety of
frequency domain bootstrap methods strongly related to the Whittle
likelihood \citep{hurvich1987frequency,franke1992bootstrapping,
  Kirch11, Kirch07, Kim13} and found important applications in a
variety of disciplines \citep{Hidalgo08, Costa13, Emmanoulopoulos13}.

Despite the fact that {\em nonparametric Bayesian} inference has been
rapidly expanding over the last decade, as reviewed by
\citet{Hjort10}, \citet{Mueller13}, and \citet{Walker13}, only very
few nonparametric Bayesian approaches to time series analysis have
been developed.  Most notably, \citet{CarterKohn97}, \citet{Gango98},
\citet{Liseo01}, \citet{Choudhuri04}, \citet{Hermansen08}, and
\citet{Chopin13} used the Whittle likelihood \citet{Whittle57} for
Bayesian modelling of the spectral density as the main nonparametric
characteristic of stationary time series.  The Whittle likelihood is
an approximation of the true likelihood. It is exact only for Gaussian
white noise.  However, even for non-Gaussian stationary time series
which are not completely specified by their first and second-order
structure, the Whittle likelihood results in asymptotically correct
statistical inference in many situations.  As shown in
\citet{Contreras06}, the loss of efficiency of the
nonparametric approach using the Whittle likelihood can be substantial
even in the Gaussian case for small samples if the autocorrelation of
the Gaussian process is high.

On the other hand, parametric methods are more powerful than
nonparametric methods if the observed time series is close to the
considered model class but fail if the model is misspecified. To
exploit the advantages of both parametric and nonparametric
approaches, the autoregressive-aided periodogram bootstrap has been
developed by \citet{Kreiss03} within the frequentist bootstrap world
of time series analysis. It fits a parametric working model to
generate periodogram ordinates that mimic the essential features of
the data and the weak dependence structure of the periodogram while a
nonparametric correction is used to capture features not represented
by the parametric fit.  This has been extended in various ways (see
\citet{Jentsch10, Jentsch12, Kreiss12}). Its main underlying idea is a
nonparametric correction of a parametric likelihood approximation. The
parametric model is used as a proxy for rough shape of the
autocorrelation structure as well as the dependency structure between
periodogram ordinates. Sensitivities with respect to the spectral
density are mitigated through a nonparametric amendment.  We propose
to use a similar nonparametrically corrected likelihood approximation
as a pseudo-likelihood in the Bayesian framework to compute the
pseudo-posterior distribution of the power spectral density (PSD) and
other parameters in time series models.  This will yield a
pseudo-likelihood that generalises the widely used Whittle likelihood
which, as we will show, can be regarded as a special case of a
nonparametrically corrected likelihood under the Gaussian i.i.d.\
working model. 
Software implementing the methodology is available in the \verb|R| 
package \verb|beyondWhittle|, which is available on
the Comprehensive R Archive Network (CRAN), see \cite{RbeyondWhittle}.

The paper is structured as follows: In Chapter~\ref{sec:likelihood},
we briefly revisit the Whittle likelihood and demonstrate that it is a
nonparametrically corrected likelihood, namely that of a Gaussian
i.i.d.\ working model.  Then, we extend this nonparametric correction
to a general parametric working model. The corresponding
pseudo-likelihood turns out to be equal to the true likelihood if the
parametric working model is correctly specified but also still yields
asymptotically unbiased periodogram ordinates if it is not correctly
specified.  In Chapter~\ref{sec:bayesian}, we propose a Bayesian
nonparametric approach to estimating the spectral density using the
pseudo-posterior distribution induced by the corrected likelihood of a
fixed parametric model.  We describe the Gibbs sampling implementation
for sampling from the pseudo-posterior. This nonparametric approach is
based on the Bernstein polynomial prior of \citet{Petrone99} and used
to estimate the spectral density via the Whittle likelihood in
\citet{Choudhuri04}.  We show posterior consistency of this approach
and discuss how to incorporate the parametric working model in the
Bayesian inference procedure.  Chapter~\ref{sec:simulations} gives
results from a simulation study, including case studies of sunspot
data, and gravitational wave data from the sixth science run of the
Laser Interferometric Gravitational Wave Observatory (LIGO).  This is
followed by discussion in Chapter~\ref{sec:summary}, which summarises
the findings and points to directions for future work. 
The proofs, the details about the Bayesian autoregressive sampler as well as 
some additional simulation results are deferred to the Appendices~\ref{sec:proofs} -- \ref{sec_appendixSims}.

\section{Likelihood approximation for time series}\label{sec:likelihood}
While the likelihood of a mean zero Gaussian time series is completely
characterised by its autocovariance function, its use for
nonparametric frequentist inference is limited as it requires
estimation in the space of positive definite covariance
functions. Similarly for nonparametric Bayesian inference, it
necessitates the specification of a prior on positive definite
autocovariance functions which is a formidable task.  A quick fix is
to use parametric models such as ARMA models with data-dependent order
selection, but these methods tend to produce biased results when the
ARMA approximation to the underlying time series is poor.
A preferable nonparametric route is to exploit
the correspondence of the autocovariance function and the spectral
density via the Wiener-Khinchin theorem and nonparametrically estimate
the spectral density. To this end, \citet{Whittle57} defined a
pseudo-likelihood, known as the Whittle likelihood, that directly
depends on the spectral density rather than the autocovariance
function and that gives a good approximation to the true Gaussian and
certain non-Gaussian likelihoods.  In the following subsection we will
revisit this approximate likelihood proposed by \citet{Whittle57},
before introducing a semiparametric approach which extends the Whittle
likelihood.

\subsection{Whittle likelihood revisited}
Assume that $\{Z_t \colon t=0, 1, \ldots\}$ is a real zero mean
stationary time series with absolutely summable autocovariance
function $\sum_{h\in \setZ}|\gamma(h)|<\infty$.  Under these
assumptions the spectral density of the time series exists and is
given by the Fourier transform (FT) of the autocovariance function
\[ 
  f(\lambda)=\frac{1}{2\pi}\sum_{k=-\infty}^{\infty} \gamma(k)\e^{-ik\lambda},\quad 0\ls \lambda\ls 2\pi.
\]
Consequently, there is a one-to-one-correspondence between the
autocovariance function and the spectral density, and estimation of
the spectral density is amenable to smoothing techniques. The idea
behind these smoothing techniques is the following observation, which
also gives rise to the so-called Whittle approximation of the
likelihood of a time series: Consider the periodogram of
$\bZ_n=(Z_1,\ldots,Z_n)^T$,
\[
	I_n(\lambda)=I_{n,\lambda}(\bZ_n)=\frac{1}{2\pi n}\left|\sum_{t=1}^nZ_t\e^{-it\lambda} \right|^2.
\]
The periodogram is given by the squared modulus of the discrete
Fourier coefficients, the Fourier transformed time series evaluated at
Fourier frequencies $\lambda_j=\frac{2\pi j}{n}$, for $
j=0,\ldots,N=\lfloor(n-1)/2 \rfloor$. It can be obtained by the
following transformation: Define for $j=1,\ldots,N$
\begin{align*}
	\mathbf{c}_j=\sqrt{2} \Re \mathbf{e}_j=\frac{1}{\sqrt{2}}(\mathbf{e}_j+\mathbf{e}_{n-j}),\qquad \mathbf{s}_j=\sqrt{2}\Im \mathbf{e}_j=\frac{1}{i\,\sqrt{2}}(\mathbf{e}_j-\mathbf{e}_{n-j}),
\end{align*}
where
\[
\mathbf{e}_j=n^{-1/2}(e_j,e_j^2,\ldots,e_j^n)^T, \qquad e_j=\exp(-2\pi i j/n),\qquad j=0,\ldots,N 
\]
and for $n$ even, $\mathbf{e}_{n/2}$ is  defined analogously.
Then, 
\begin{align*}
		F_n=\begin{cases}
		(\mathbf{e}_0,\mathbf{c}_1,\mathbf{s}_1,\ldots,\mathbf{c}_N,\mathbf{s}_N,\mathbf{e}_{n/2})^T, & n\text{ even},\\
		(\mathbf{e}_0,\mathbf{c}_1,\mathbf{s}_1,\ldots,\mathbf{c}_N,\mathbf{s}_N)^T, & n\text{ odd},
	\end{cases}
\end{align*}
is an orthonormal $n\times n$ matrix (cf. e.g.\
\citet{brockwell2009time}, paragraph 10.1).  Real- and imaginary parts
of the discrete Fourier coefficients are collected in the vector
\begin{align*}
{\tilde{\mathbf Z}}_n:=(\tilde{Z}_n(0),\ldots,\tilde{Z}_n(n-1))^T=F_n \bZ_n 
\end{align*}
and the periodogram can be written as
\begin{align}\label{eq_perio_four}
	&I_n(\lambda_j)=\frac{1}{4\pi}\left(\tilde{Z}_n^2(2j)+\tilde{Z}^2_n(2j-1)\right), \quad j=1,\ldots,N,\notag\\
	&I_n(\lambda_0)=\frac{1}{2\pi}\tilde{Z}_n^2(0),\qquad \text{as well as for } n\text{ even: } I_{n}(\lambda_{n/2})=\frac{1}{2\pi}\tilde{Z}_n^2(n-1).
\end{align}
It is well known that the periodograms evaluated at two different
Fourier frequencies are asymptotically independent and have an
asymptotic exponential distribution with mean equal to the spectral
density, a statement that remains true for non-Gaussian and even
non-linear time series \citet{Shao07}.  Similarly, the Fourier
coefficients $\boldsymbol{\tilde{Z}}_n$ are asymptotically independent
and normally distributed with variances equal to $2\pi$ times the
spectral density at the corresponding frequency. This result gives
rise to the following Whittle approximation in the frequency domain
\begin{equation}\label{eq_NormalFouriercoefficients}
	p_{\widetilde{\Z}_n}^W(\boldsymbol{\tilde{z}}_n|f) \propto\det(D_n)^{-1/2}\exp\left(-\frac 1 2 \boldsymbol{\tilde{z}}_n^TD_n^{-1}\boldsymbol{\tilde{z}}_n\right)
\end{equation}
by the likelihood of a Gaussian vector with diagonal covariance matrix
\begin{equation}\label{eq_def_Dn}
	D_n:=D_n(f):=2\pi\,\begin{cases}
	\text{diag}(f(0),f(\lambda_1),f(\lambda_1),\ldots,f(\lambda_N),f(\lambda_N),f(\lambda_{n/2}))& n \text{ even},\\
	\text{diag}(f(0),f(\lambda_1),f(\lambda_1),\ldots,f(\lambda_N),f(\lambda_N))& n \text{ odd}.
\end{cases}
\end{equation}
As explicitly shown in Appendix~\ref{sec:proofs},
this yields the famous Whittle
likelihood in the time domain via the transformation theorem
\begin{align}\label{eq_whittle}
	p^W_{\mathbf{Z}_n=F_n^T\boldsymbol{\tilde{Z}}_n}(\boldsymbol z_n|f) \propto\exp\left\{ - \frac 1 2 \sum_{j=0}^{n-1}\left( \log f(\lambda_j)+\frac{I_{n,\lambda_j}(\boldsymbol z_n)}{f(\lambda_j)} \right) \right\}
\end{align}
which provides an approximation of the true likelihood. It is exact
only for Gaussian white noise in which case
$f(\lambda_j)=\sigma^2/2\pi$. It has the advantage that it depends
directly on the spectral density in contrast to the true likelihood
that depends on $f$ indirectly via Wiener-Khinchin's
theorem. Sometimes, the summands corresponding to $j=0$ as well as
$j=n/2$ (the latter for $n$ even) are omitted in the likelihood
approximation. In fact, the term corresponding to $j=0$ contains the
sample mean (squared) while the term corresponding to $j=n/2$ gives
the alternating sample mean (squared). Both have somewhat different
statistical properties and usually need to be considered
separately. Furthermore, the first term is exactly zero if the methods
are applied to time series that have been centered first, while the
last one is approximately zero and asymptotically negligible (refer
also Remark~\ref{rem_firstlast}).

The density of~$F_n^TD_n^{1/2}F_n\Z_n$ under the i.i.d. standard
Gaussian working model is the Whittle likelihood.  It has two
potential sources of approximation errors: The first one is the
assumption of independence between Fourier coefficients which holds
only asymptotically but not exactly for a finite time series, the
second one is the Gaussianity assumption. In this paper, we restrict
our attention to the first problem, extending the proposed methods to
non-Gaussian situations will be a focus of future work. In fact, the
independence assumption leads to asymptotically consistent results for
Gaussian data.  But even for Gaussian data with relatively small
sample sizes and relatively strong correlation the loss of efficiency
of the nonparametric approach using the Whittle likelihood can be
substantial as shown in \citet{Contreras06} or by the simulation
results of \citet{Kreiss03}.

\subsection{Nonparametric likelihood correction}
The central idea in this work is to extend the Whittle likelihood by
proceeding from a certain parametric working model (with mean 0) for
$\Z_n$ rather than an i.i.d.\ standard Gaussian working model before
making a correction analogous to the Whittle correction in the
frequency domain.

To this end, we start with some parametric likelihood in the time
domain, such as e.g.\ obtained from an ARMA-model, that is believed to
be a reasonable approximation to the true time series. We denote the
spectral density that corresponds to this parametric working model by
$f_{\text{param}}(\cdot)$. If the model is misspecified, then this
spectral density is also wrong and needs to be corrected to obtain the
correct second-order dependence structure. To this end, we define a
correction matrix
\begin{align*}
	&C_n=C_n(f,f_{\text{param}})=C_n(c(\lambda))\qquad \left(c:=\frac{f}{f_{\text{param}}}\right)\\
	&=\begin{cases}
	\text{diag}\left(\frac{f(\lambda_0)}{f_{\text{param}}(\lambda_0)},\frac{f(\lambda_1)}{f_{\text{param}}(\lambda_1)},\frac{f(\lambda_1)}{f_{\text{param}}(\lambda_1)},\ldots,\frac{f(\lambda_N)}{f_{\text{param}}(\lambda_N)},\frac{f(\lambda_N)}{f_{\text{param}}(\lambda_N)},\frac{f(\lambda_{n/2})}{f_{\text{param}}(\lambda_{n/2})}\right)& n \text{ even},\\
	\text{diag}\left(\frac{f(\lambda_0)}{f_{\text{param}}(\lambda_0)},\frac{f(\lambda_1)}{f_{\text{param}}(\lambda_1)},\frac{f(\lambda_1)}{f_{\text{param}}(\lambda_1)},\ldots,\frac{f(\lambda_N)}{f_{\text{param}}(\lambda_N)},\frac{f(\lambda_N)}{f_{\text{param}}(\lambda_N)}\right)& n \text{ odd}.
\end{cases}
\end{align*}
This is analogous to the Whittle correction in the previous section
as, in particular, $D_n=C_n(f,f_{\text{ i.i.d.} N(0,1)})$ with $D_n$
as in \eqref{eq_def_Dn}. However, the corresponding periodogram
ordinates are no longer independent under this likelihood but instead
inherit the dependence structure from the original parametric model
(see Proposition \ref{cor_31} c). Such an approach in a bootstrap
context has been proposed and successfully applied by \citet{Kreiss03}
using an AR($p$) approximation.  This concept of a nonparametric
correction of a parametric time domain likelihood is illustrated in
the schematic diagram:
\[
\begin{array}{lcl}
\underline{\mbox{time domain}} & & \underline{\mbox{frequency domain}}\\
\Z_n \sim \text{parametric working model} &\mbox{ } \stackrel{\mbox{ $\operatorname{FT}$}}{\xrightarrow{\hspace*{2.5cm}}} \mbox{ }& F_n\Z_n\\
& & \xdownarrow{0.8cm} C_n(f,f_{\text{param}})\\
F_n^TC_n(f,f_{\text{param}})^{1/2}F_n\Z_n &\stackrel{\mbox{ $\operatorname{FT}^{-1}$}}{\xleftarrow{\hspace*{2.5cm}}}& C_n(f,f_{\text{param}})^{1/2}F_n\Z_n\\
\end{array}\]
As a result we obtain the following nonparametrically corrected likelihood
function under the parametric working model
\begin{align}\label{eq_corrected_param}
	p_{\text{param}}^C(\Z_n|f)\propto \det(C_n)^{-1/2}\; p_{\text{param}}(F_n^TC_n^{-\frac{1}{2}} F_n\Z_n),\end{align}
where $p_{\text{param}}$ denotes the parametric likelihood. 
\begin{remark}\label{rem_identifiability}
  Parametric models with a multiplicative scale parameter $\sigma\neq
  1$ yield the same corrected likelihood as the one with $\sigma=1$ ,
  i.e.\ if $\sigma \Z_n$ is used as working model this leads to the
  same corrected likelihood for all $\sigma>0$.  For instance, if the
  parametric model is given by i.i.d.\ $N(0,\sigma^2)$ random
  variables with $\sigma^2>0$ arbitrary, then the correction also
  results in the Whittle likelihood (for a proof we refer to
  Appendix~\ref{sec:proofs}).  
  Analogously, for linear models
  $Z_t=\sum_{l=-\infty}^{\infty}d_le_{t-l}$, $e_t\sim (0,\sigma^2)$,
  which includes the class of ARMA-models, the corrected likelihood is
  independent of~$\sigma^2$.
\end{remark}

We can now prove the following proposition which shows two important
things: First, the corrected likelihood is the exact likelihood in
case the parametric model is correct.  Second, the periodograms
associated with this likelihood are asymptotically unbiased for the
true spectral density regardless of whether the parametric model is
true.

\begin{proposition}\label{cor_31}
  Let $\{Z_t\}$ be a real zero mean stationary time series with
  absolutely summable autocovariance function
  $\sum_{h\in\mathbb{Z}}|\gamma(h)|<\infty$ and let
  $f_{\text{param}}(\lambda)\gs \beta >0$ for $0\ls \lambda\ls \pi$ be
  the spectral density associated with the (mean zero) parametric
  model used for the correction.
	\begin{enumerate}
			\item If $f=f_{\text{param}}$, then $p_{\text{param}}^C=p_{\text{param}}$.
		\item The periodogram associated with the corrected likelihood is asymptotically unbiased for the true spectral density, i.e.\
			\begin{align*}
				\E_{p_{\text{param}}^C} I_{n,\lambda_j}(\Z_n)&=\int I_{n,\lambda_j}(z_1,\ldots,z_n)\,d p_{\text{param}}^C(z_n,\ldots,z_n)\\
				&{=\frac{f(\lambda_j)}{f_{\text{param}}(\lambda_j)} \E_{p_{\text{param}}}I_{n,\lambda_j}(\Z_n)}=f(\lambda_j)+o(1),
			\end{align*}
			where the convergence is uniform in $j=0,\ldots, \lfloor (n-1)/2\rfloor$. Furthermore, 
				\begin{align*}
					\cov_{p_{\text{param}}^C}(I_{n,\lambda_l}(\Z_n),I_{n,\lambda_k}(\Z_n))=\frac{f(\lambda_l)f(\lambda_k)}{f_{\text{param}}(\lambda_l)\,f_{\text{param}}(\lambda_k)}\,\cov_{p_{\text{param}}}(I_{n,\lambda_l}(\Z_n),I_{n,\lambda_k}(\Z_n)).
				\end{align*}
	\end{enumerate}
      \end{proposition} {The proof shows that the vector of
        periodograms under the corrected likelihood has exactly the
        same distributional properties as the vector of the
        periodograms under the parametric likelihood multiplied with
        $f(\cdot)/f_{\text{param}}(\cdot)$. Hence, asymptotic
        properties as the ones derived in Theorem 10.3.2 in
        \citet{brockwell2009time} carry over with the appropriate
        multiplicative correction.}

      In the remainder of the paper we describe how to make use of
      this nonparametric correction in a Bayesian set-up.

\section{Bayesian semiparametric approach to time series analysis}\label{sec:bayesian}
To illustrate the Bayesian semiparametric approach and how to sample
from the pseudo-posterior distribution, in the following we restrict
our attention to an AR($p$) model as our parametric working model for
the time series, i.e.  $Z_i=\sum_{l=1}^pa_lZ_{i-l}+\eps_i$, where
$\{\eps_i\}$ are i.i.d. N($0,1)$ random variables with density denoted
by $\varphi(\cdot)$. Note that without loss of generality,
$\sigma^2=1$, cf.\ Remark~\ref{rem_identifiability}.  This yields the
parametric likelihood of our working model, depending on the order~$p
\geq 0$ and on the coefficients~$\mathbf a = (a_1, ..., a_p)$:
\begin{equation}\label{eq_param_dens}
  p_{\text{param}}(\Z_n|\mathbf a)
  \propto 
    p_{\text{param}}(\Z_p|\mathbf a)
    \prod_{i=p+1}^n \varphi\left(Z_i - \sum_{l=1}^pa_lZ_{i-l}\right)
\end{equation}
with spectral density
\begin{equation}\label{eq_def_sd_param}
  f_{\text{param}}(\lambda ; \mathbf a)=\frac{1}{2\pi}\left|1-\sum_{l=1}^p a_l e^{-il\lambda}\right|^{-2}.
\end{equation}
We assume the time series to be stationary and causal a priori.  Thus,
$\mathbf a$ is restricted such that~$\phi(z):=1-a_1z-...-a_pz^p$ has
no zeros inside the closed unit disc, c.f. Theorem 3.1.1. in
\citet{brockwell2009time}.  For now, we assume that the
parameters~$(p,\mathbf a)$ of the parametric working model are fixed
(and in practice set to Bayesian point estimates obtained from a
preceding parametric estimation step).  An extension to combine the
estimation of the parametric model with the nonparametric correction
will be detailed later in Section~\ref{sec_priorParametric}.

\subsection{Nonparametric prior for spectral density inference}\label{sec_bernsteinDirichlet}

For a Bayesian analysis using either the Whittle or nonparametrically
corrected likelihood, we need to specify a nonparametric prior
distribution for the spectral density. Here we employ the approach by
\citet{Choudhuri04} which is essentially based on the Bernstein
polynomial prior of \citet{Petrone99} as a nonparametric prior for a
probability density on the unit interval. We briefly describe the
prior specification and refer to \citet{Choudhuri04} for further
details.

In contrast to the approach in \citet{Choudhuri04}, we do not specify
a nonparametric prior distribution for the spectral
density~$f(\cdot)$, but for a pre-whitened version thereof,
incorporating the spectral density of the parametric working model
into the estimation.  To elaborate, for~$0\leq \eta \leq 1$, consider
the \emph{eta-damped correction function}
\begin{equation}\label{eq_etaCorrection}
c_{\eta}(\lambda)=  c_{\eta}(\lambda;{\mathbf a}) := f(\lambda) / f_{\text{param}}(\lambda;{\mathbf a})^\eta.
\end{equation}
This corresponds to a reparametrization of the likelihood~\eqref{eq_corrected_param} by replacing  
$C_n=C_n(c(\lambda))$ with $C_n=C_n( c_\eta(\lambda;{\mathbf
    a})f_{\text{param}}(\lambda; {\mathbf a})^{\eta-1} )$.

\begin{remark}\label{rem_etaCorrection}
  The parameter $\eta$ models the confidence in the parametric model:
  If $\eta$ is close to~$1$ and the model is well-specified,
  then~$c_\eta(\cdot)$ will be much smoother than the original
  spectral density, since $f_{\text{param}}(\cdot)$ already captures
  the prominent spectral peaks of the data very well.  As a
  consequence, nonparametric estimation of~$c_{\eta}(\cdot)$ should
  involve less effort than nonparametric estimation of~$f(\cdot)$
  itself.  This remains true in the misspecified case, as long as the
  parametric model does describe the essential features of the data
  sufficiently well in the sense that it captures at least the more
  prominent peaks.  However, it is possible that the parametric model
  introduces erroneous spectral peaks if the model is misspecified.
  In that case, $\eta$ close to zero ensures a damping of the model
  misspecification, such that nonparametric estimation
  of~$c_{\eta}(\cdot)$ should involve less effort than nonparametric
  estimation of~$f(\cdot)/f_{\text{param}}(\cdot)$.  The choice of
  $\eta$ will be detailed in Section~\ref{sec_priorParametric}, but
  for now, $\eta$ is assumed fixed.
\end{remark}
We reparametrise $c_{\eta}(\cdot)$ to a density function $q(\cdot)$ on
$[0,1]$ via $c_\eta(\pi\omega)=\tau q(\omega), 0\ls \omega\ls 1$ with
normalization constant $\tau=\int_0^1 c_\eta(\pi\omega)d\omega$. Thus,
a prior for $c_\eta(\cdot)$ may be specified by putting a Bernstein
polynomial prior on $q(\cdot)$ and then an independent
Inverse-Gamma$(\alpha_\tau,\beta_\tau)$ prior on $\tau$, its density
denoted by $p_\tau$.  The Bernstein polynomial prior of $q$ is
specified in a hierarchical way as follows:
\begin{enumerate}
\item $\displaystyle q(\omega|k,G)=\sum_{j=1}^k G\left(
    \frac{j-1}{k},\frac{j}{k}\right] \beta(\omega|j,k-j+1)$ where
  $G(u,v]=G(v)-G(u)$ for a distribution function $G$ and
  $\beta(\omega|l,m)=\frac{\Gamma(l+m)}{\Gamma(l)\Gamma(m)}\;
  \omega^{l-1}(1-\omega)^{m-1}$ is the beta density with parameters
  $l$ and $m$.
\item $G$ has a Dirichlet process distribution with base measure
  $\alpha=M G_0$, where $M>0$ is a constant and $G_0$ a distribution
  function with Lebesgue density $g_0$.
\item $k$ has a discrete distribution on the integers $k=1,2,\ldots$,
  independent of $G$, with probability function $p_k(k)\propto
  \exp(-\theta_k k\log(k))$. Note that smaller values of $k$ yield
  smoother densities.
\end{enumerate}
Furthermore, we achieve an approximate finite-dimensional
characterization of this nonparametric prior in terms of $2L+3$
parameters $(V_1,\ldots,V_L,W_0,W_1,\ldots,W_L,k,\tau)$ by employing
the truncated Sethuraman (1994) representation of the Dirichlet
process
\[G=\sum_{l=1}^L p_l \delta_{W_l} +(1-p_1-\ldots -p_L)\delta_{W_0}\]
with $p_1=V_1$, $p_l=\left(\prod_{j=1}^{l-1}(1-V_{j})\right)V_l$ for
$l\gs 2$, $V_l\sim \mbox{beta}(1,M)$, and $W_l\sim G_0$, all
independent.  This gives a prior finite mixture representation of the
eta-damped correction
\begin{equation}\label{eq_pirormixture}
  c_\eta(\pi\omega) = \tau \sum_{j=1}^k \tilde{w}_{j,k} \beta(\omega|j,k-j+1),
\end{equation}
where $\tilde{w}_{j,k}=\sum_{l=0}^L p_l I\{\frac{j-1}{k}<W_l\ls \frac{j}{k}\}$ and $p_0=1-\sum_{l=1}^L p_l$.

The joint prior density of~$c_\eta$ by means of this
finite-dimensional approximation can be written as
\begin{equation}\label{eq_prior}
  p(V_1,\ldots,V_L,W_0,W_1,\ldots,W_L,k,\tau)\propto \left(\prod_{l=1}^L M(1-V_l)^{M-1}\right)
  \left(\prod_{l=0}^L g_0(W_l)\right) p_k(k)\,p_\tau(\tau).
\end{equation} 

Here, we specify a diffuse prior by choosing the uniform distribution
for $G_0$ and $M=1$. We set $\theta_k=0.01$,
$\alpha_\tau=\beta_\tau=0.001$ and follow the recommendation by
\citet{Choudhuri04} for the truncation point $L=\max\{20,n^{1/3}\}$.

\subsection{Posterior computation}\label{sec_posteriorComputation}
The prior~\eqref{eq_pirormixture} on $c_\eta(\cdot)$ induces a prior
on $f(\cdot)$ by multiplication with~$f_{\text{param}}(\cdot;\mathbf
a)^\eta$, see \eqref{eq_etaCorrection}.  Accordingly, the
pseudo-posterior distribution of $f(\cdot)$ can be computed as prior
times the corrected parametric likelihood:
\begin{eqnarray*}
&&p_{\text{post}}^{C}(V_1,\ldots,V_L,W_0,W_1,\ldots,W_L,k,\tau|\mathbf{Z}_n,\mathbf a, \eta)\\
&&\qquad\propto p(V_1,\ldots,V_L,W_0,W_1,\ldots,W_L,k,\tau) \det(C_n)^{-1/2} p_{\text{param}}(F_n^TC_n^{-1/2}F_n\mathbf{Z}_n|\mathbf a),
\end{eqnarray*}
where~$C_n = C_n\left(c_\eta(\lambda;{\mathbf
    a})f_{\text{param}}(\lambda; {\mathbf a})^{\eta-1}\right)$
and~$f_{\text{param}}(\lambda; \mathbf a)$ as
in~\eqref{eq_def_sd_param}.  Samples from the pseudo-posterior
distribution can be obtained via Gibbs sampling following the steps
outlined in \citet{Choudhuri04}.  The full conditional for $k$ is
discrete and readily sampled, as is the conjugate full conditional of
$\tau$. We use the Metropolis algorithm to sample from each of the
full conditionals of $V_l$ and $W_l$ using the uniform proposal
density of \citet{Choudhuri04}.
\begin{remark}\label{rem_firstlast}
  As in \citet{Choudhuri04}, we omit the first and last terms in the
  corrected likelihood that correspond to $\lambda=0$ and
  $\lambda=n/2$ (and setting $c_\eta(0)=c_\eta(n/2)=0$ as well as
  $F_n\Z_n(0)=F_n\Z_n(n)=0$).  This is due to the role that the
  corresponding Fourier coefficients play (being equal to the sample
  mean respectively alternating sample mean), which typically requires
  a special treatment (see Proposition 10.3.1 and (10.4.7) in
  \citet{brockwell2009time}). For the application to spectral density
  estimation in this paper this leads to more stable statistical
  procedures irrespective of the true mean of the time
  series. However, in situations, where the time series is merely used
  as a nuisance parameter such as regression models, change point or
  unit-root testing, these coefficients should be included and the
  likelihood used for the time series $Z_t-\mu$, where $\mu$ is the
  mean (not the sample mean) of the time series.
\end{remark}

\subsection{Posterior consistency}
In this section, we will show consistency of the pseudo-posterior
distribution based on the Bernstein polynomial prior and the corrected
likelihood for a given working model under the same assumptions as
\citet{Choudhuri04}. Throughout the section, we will make the following
assumption:
\begin{ass}\label{ass_1}
	\begin{enumerate}
        \item Denote by $\gamma_{\text{param}}(\cdot)$ and
          $f_{\text{param}}(\cdot)$ the autocovariance function
          respectively spectral density of the parametric working
          model. Assume that
			\begin{align*}
				&		\sum_{h\in \mathbb{Z}}h^{\alpha}|\gamma_{\text{param}}(h)|<\infty \text{ for some }\alpha>1,\qquad\\
				&f_{\text{param}}(\lambda)\gs \beta>0 \text{ for some }\beta>0\quad\text{ and all }-\pi\ls \lambda\ls \pi. 
			\end{align*}
                      \item Let $\{Z_t\}$ be a stationary mean zero
                        Gaussian time series with autocovariance
                        function $\gamma_0(\cdot)$ and spectral
                        density $f_0(\cdot)$ fulfilling
	\begin{align*}
		&		\sum_{h\in \mathbb{Z}}h^{\alpha}|\gamma_0(h)|<\infty \text{ for some }\alpha>1,\qquad\\
		&f_0(\lambda)\gs \beta>0 \text{ for some }\beta>0\quad\text{ and all }-\pi\ls \lambda\ls \pi. 
			\end{align*}
                        Denote by~$p_{n,f_0}(\cdot)$ and~$P_{n,f_0}$
                        the density and the distribution of~$\Z_n =
                        (Z_1,\ldots,Z_n)$.
\end{enumerate}
\end{ass}

An important first observation is, that the corrected likelihood, the
Whittle likelihood as well as the true likelihood are all mutually
contiguous in the Gaussian case. This fact may also be of independent
interest:
\begin{theorem}\label{theorem_contiguity}
  Under Assumptions~$\mathcal{A}$.\ref{ass_1} the true density
  $p_{n,f_0}(\cdot)$, the Whittle likelihood $p^W(\cdot|f)$ given in
  \eqref{eq_whittle} as well as the corrected (Gaussian) parametric
  likelihood $p_{\text{param}}^C(\cdot|f)$ given
  in~\eqref{eq_corrected_param} are all mutually contiguous.
\end{theorem}

With the help of this theorem we are now able to prove posterior
consistency under certain assumptions on the time series and prior.
\begin{theorem}\label{theorem:consistency-corrected}
  Let~$0 \leq \eta \leq 1$ fixed.  Let
  Assumptions~$\mathcal{A}$.\ref{ass_1} hold in addition to the
  following assumptions on the prior for~$c_\eta(\cdot)$:
\begin{itemize}
\item[(P1)] for all $k$, $0<p_k(k)\ls B\e^{-b\, k \log k}$ for some
  constants $B,b>0$,
\item[(P2)] $g_0$ is bounded, continuous and bounded away from 0,
\item[(P3)] the parameter $\tau$ is assumed fixed and known.
\end{itemize}
Let~$c_{0,\eta}(\lambda) = f_0(\lambda)/f_{\text{param}}^\eta(\lambda)$.
Then the posterior distribution is consistent, i.e.\ for any $\epsilon>0$,
\[\Pi_n(c: ||c-c_{0,\eta}||_1>\epsilon|\Z_n) \rightarrow 0\]
in $P_{n,f_0}$-probability, where $\Pi_n(\cdot|\Z_n)$ denotes the
pseudoposterior distribution computed using the corrected likelihood.
\end{theorem}

\subsection{Prior for the parameters of the working model}\label{sec_priorParametric}
In the previous sections, the parameters of the working model were
assumed to be fixed, as e.g. obtained in an initial pre-estimation
step.  From a Bayesian perspective, it is desirable to couple the
inference about the parametric working model with the nonparametric
correction, allowing for the inclusion of prior knowledge about the
model and for uncertainty quantification about the interaction of
model and correction.  Thus, for a fixed order~$p$, we include both
the autoregressive parameters~$\mathbf a = (a_1,\ldots,a_p)$ and the
spectral shape confidence~$\eta$ from~\eqref{eq_etaCorrection} into
the Bayesian inference.  {\cb The introduction of the parameter $\eta$  effectively robustifies the procedure in the sense that it guarantees
our method will not be worse than a corresponding fully nonparametric one. }

To ensure stationarity and causality (and
hence identifiability) of the parametric model, we put a prior on the
partial autocorrelations~$\brho = (\rho_1, \ldots, \rho_p)$
with $-1 < \rho_l < 1$ for $1 \leq l \leq p$.  The autoregressive
parameters~$\mathbf a = \mathbf a(\brho)$ can be readily
obtained from this parametrisation (see Appendix~\ref{sec_appendixSampling}).

We consider the following prior specification for the spectral density:
\begin{align*}
  f(\lambda)
  &= c_\eta(\lambda) f_{\text{param}}(\lambda; \brho)^\eta,
\end{align*}
with a Bernstein-Dirichlet prior on~$c_\eta(\cdot)$ as in
Section~\ref{sec_bernsteinDirichlet}, a uniform prior on $\eta$ and
uniform priors on the $\rho_l$'s, all a priori independent.  Of
course, it is possible to employ different prior models (see
\citet{liseo2013objective, sorbye2016penalised}).  In conjunction with
the corrected parametric likelihood, we obtain samples from the joint
pseudo-posterior distribution
\[
  p_{\text{post}}^{C}(v_1,\ldots,v_L,w_0,w_1,\ldots,w_L,k,\tau,\rho_1,\ldots,\rho_p,\eta|\mathbf{Z}_n)
\]
analogously to Section~\ref{sec_posteriorComputation} via Gibbs
sampling.  Note that, since the corrected parametric likelihood is the
Lebesgue density of a probability measure, it is sufficient that the
prior distributions are proper for the posterior distribution to be
proper.  We use random walk Metropolis-within-Gibbs steps with normal
proposal densities to sample from the full conditionals of~$\eta$
and~$\rho_1, ..., \rho_p$ respectively.  The proposal variance
for~$\eta$ is set to~0.01, where proposals larger than 1 (smaller than
0) are truncated at 1 (at 0).  To achieve proper mixing of the
parametric model parameters, the proposal variances~$\sigma_l^2$
for~$\rho_l$ are determined adaptively as described in
\citet{roberts2009examples} during the burn-in period, aiming for an
acceptance rate of~0.44, where proposals with absolute value larger or
equal to one are discarded.

\begin{remark}\label{rem_arOrder}
  The autoregressive order~$p$ is assumed to be fixed.  In our
  approach, it is determined in a preliminary model selection step.
  However, it is also possible to include the autoregressive order in
  the Bayesian inference by using a Reversible-jump Markov Chain Monte
  Carlo scheme \citet{green1995reversible} or stochastic search
  variables \citet{barnett1996bayesian}.
\end{remark}

\section{Numerical evaluation}\label{sec:simulations}
In this section, we evaluate the finite sample behavior of our
\emph{nonparametrically corrected (NPC)} approach to Bayesian spectral
density estimation numerically. To demonstrate the
trade-off between the parametric working model and the nonparametric
spectral correction, we compare our approach to both fully
parametric and fully nonparametric approaches.
We first present the results of a simulation study with ARMA data
in Section~\ref{sec:simArma} before considering sunspot data in
Section~\ref{sec:sunspot} and gravitational wave data in Section~\ref{sec:ligo}.
An implementation of all procedures presented below is provided
in the \verb|R| package \verb|beyondWhittle|, 
which is available on CRAN, see \cite{RbeyondWhittle}.

\subsection{Simulated ARMA data}\label{sec:simArma}
We consider data generated from the ARMA model
\begin{equation} \label{eq_model_ar1}
  Z_t = aZ_{t-1} + be_{t-1} + e_t, \quad 1 \leq t \leq n
\end{equation}
with standard Gaussian white noise~$e_t$ and different values of~$a,b$
and~$n$.  
The following competing approaches are compared with NPC:
\par
\emph{Nonparametric estimation (NP).}  The procedure from
\citet{Choudhuri04}, which is based on the Whittle likelihood and a
Bernstein-Dirichlet prior on the spectral density.  Note that this
coincides with the NPC approach with a white noise parametric working
model ($p=0$), c.f. Remark~\ref{rem_identifiability}.

\par
\emph{Autoregressive estimation (AR).}
For~$p=0,1,\ldots,p_{\text{max}}$, an autoregressive model of
order~$p$ is fitted to the data using a Bayesian approach with the
same partial autocorrelation parametrization and the same prior
assumptions as for the parametric working model within the
nonparametrically corrected likelihood procedure, see
Section~\ref{sec_priorParametric} (for details on the sampling scheme
we refer to Appendix~\ref{sec_appendixSampling}). The order~$p^*$
minimizing the DIC is then chosen for model comparison.

\par
The working model in the NPC approach is chosen to be the~AR($p^*$)
model from the AR procedure.  The prior for the working model
parameters is as described in Section~\ref{sec_priorParametric} and
the prior for the nonparametric correction is as described in
Section~\ref{sec_bernsteinDirichlet}.  For the NP approach, the same
Bernstein-Dirichlet prior for~$f(\cdot)$ is used as
for~$c_\eta(\cdot)$ in the NPC approach.

\par
We compare the average Integrated Absolute Error (aIAE) of
the posterior median spectral density estimate and the empirical
coverage probability of a Uniform Credible Interval (cUCI).  Note that
pointwise posterior credible intervals are not suited for
investigating coverage, since they do not take the multiple testing
problem into account that arises at different frequencies.  Following
\citet{haefnerKirch2016} (see also \citet{neumann1998simultaneous}), a
Uniform Credible Interval for the spectral density can be constructed
as follows: Denote by~$f^*_1(\cdot),\ldots, f^*_N(\cdot)$ the
posterior spectral density samples obtained by one of the procedures.
Then for~$0 < \alpha < 1$ the Uniform $\alpha$-Credible Interval is
given by
\[
  \left[
    \exp \left( \log f^*(\lambda) - C^*_\alpha  \sigma^*(\lambda) \right),
    \exp \left( \log f^*(\lambda) + C^*_\alpha  \sigma^*(\lambda) \right)
  \right]
\]
where~$f^*(\lambda)$ denotes the sample median at frequency~$0 \leq
\lambda \leq \pi$, $\sigma^*(\lambda)$ the median absolute deviation
of~$\log f^*_1(\lambda),\ldots,\log f^*_N(\lambda)$ and~$C^*_\alpha$
is chosen such that
\[
  \frac{1}{N} \sum_{j=1}^N 1 \left\{ 
    \max_{0 \leq \lambda \leq \pi} \frac{\log f^*(\lambda) - \log f^*_j(\lambda)}{\sigma^*(\lambda)} \leq C^*_\alpha
  \right\}
  \geq 1 - \alpha.
\]
The intervals are constructed on a logarithmic scale to ensure that
their covered range contains only positive values.  Because small
values of~$f^*$ lead to very large absolute values on a log scale, we
do not employ the usual logarithm, but the Fuller-logarithm as
described in~\citet{fuller1996introduction}, page 496, i.e.
\[
  \log_{\operatorname{Fuller}}(x) = \log(x + \xi) - \xi / (x + \xi),
\]
for some small~$\xi > 0$.  We use~$\alpha=0.9$ and $\xi=0.001$ in our
simulations.  
The chains were run for 12,000 iterations for AR (after a burn-in period of 8,000 iterations) and
for 20,000 iterations for NP and NPC (after a burn-in period of 30,000 iterations),
where a thinning of~4 was employed for NP and NPC.
We choose~$p_{\text max}=15$ and consider lengths~$n=64,128,256$ from
model~\eqref{eq_model_ar1} with~$N=1024$
replicates
($N$ a power of~2 to use the computational resources efficiently)
respectively.

\par
The results are shown in Table~1. 
For~AR(1) data, the AR
procedure yields the best results (in terms of both aIAE and cUCI),
whereas the NP performs worst.  It can be seen that AR and NPC benefit
from the well-specified parametric model.  For~MA(1) data, however,
the AR approach yields the worst results, whereas NPC benefits from
the nonparametric correction, yielding only slighly worse integrated
errors than NP, although with superior uniform credible intervals
for~$n \geq 128$.  In case of ARMA(1,1) data, the estimation does not
benefit from the autoregressive fit, i.e. the moving average
misspecification dominates the estimation.  Thus the results are
similar to the MA(1) case.  Further results for data from the ARMA
model can be found in Appendix~\ref{sec_appendixSims}.
\begin{table}
\caption{Average Integrated Absolute Error (aIAE), Uniform 0.9-Credible Interval coverage (cUCI) and
average posterior model confidence~$\hat \eta$ for
different realizations of model~\eqref{eq_model_ar1}.}\label{tab_iae1}
\begin{adjustbox}{max width=\linewidth}
\fbox{
\begin{tabular}{l|rrrlrrrlrrr}
&          \multicolumn{3}{c}{AR($1$): $a=0.95$}         & & \multicolumn{3}{c}{MA($1$): $b=0.8$}        & & \multicolumn{3}{c}{ARMA($1,1$): $a=0.75$, $b=0.8$} \\
\cline{2-4}\cline{6-8}\cline{10-12}
&   $n=64$  &  $n=128$  &  $n=256$    & & $n=64$  &  $n=128$  &  $n=256$   & & $n=64$  &  $n=128$ & $n=256$  \\
\hline 
aIAE &&&&&&&&&&&\\
AR&       2.661  &  2.101   &  1.600    & & 0.298  &  	0.244   &  0.192  & & 1.236  &  1.038  & 0.862  \\
NP&       3.543  &  2.946   &  2.370    & & 0.197  &  	0.153   &  0.119  & & 1.022  &  0.806  & 0.625  \\
NPC&      2.992  &  2.240   &  1.612    & & 0.206  &  	0.157   &  0.121  & & 1.083  &  0.907  & 0.727  \\
\hline
cUCI &&&&&&&&&&& \\
AR&       0.948  &  0.963   &  0.984    & & 0.876  &  	0.860   &  0.867  & & 0.866  &  0.846  & 0.891  \\
NP&       0.863  &  0.771   &  0.801    & & 0.998  &  	0.999   &  0.995  & & 0.953  &  0.919  & 0.906  \\
NPC&      0.952  &  0.973   &  0.996    & & 0.998  &  	1.000   &  1.000  & & 0.999  &  1.000  & 0.998  \\
\hline
$\hat \eta$ & 0.697  &  0.818   &  0.896    & & 0.285  &   0.191   &  0.121  & & 0.483  &  0.384  & 0.272
\end{tabular}}
\end{adjustbox}
\end{table}

\begin{figure}[b]
 \centering
  \subfigure[]{
    \centering
 	\includegraphics[width=.22\textwidth]{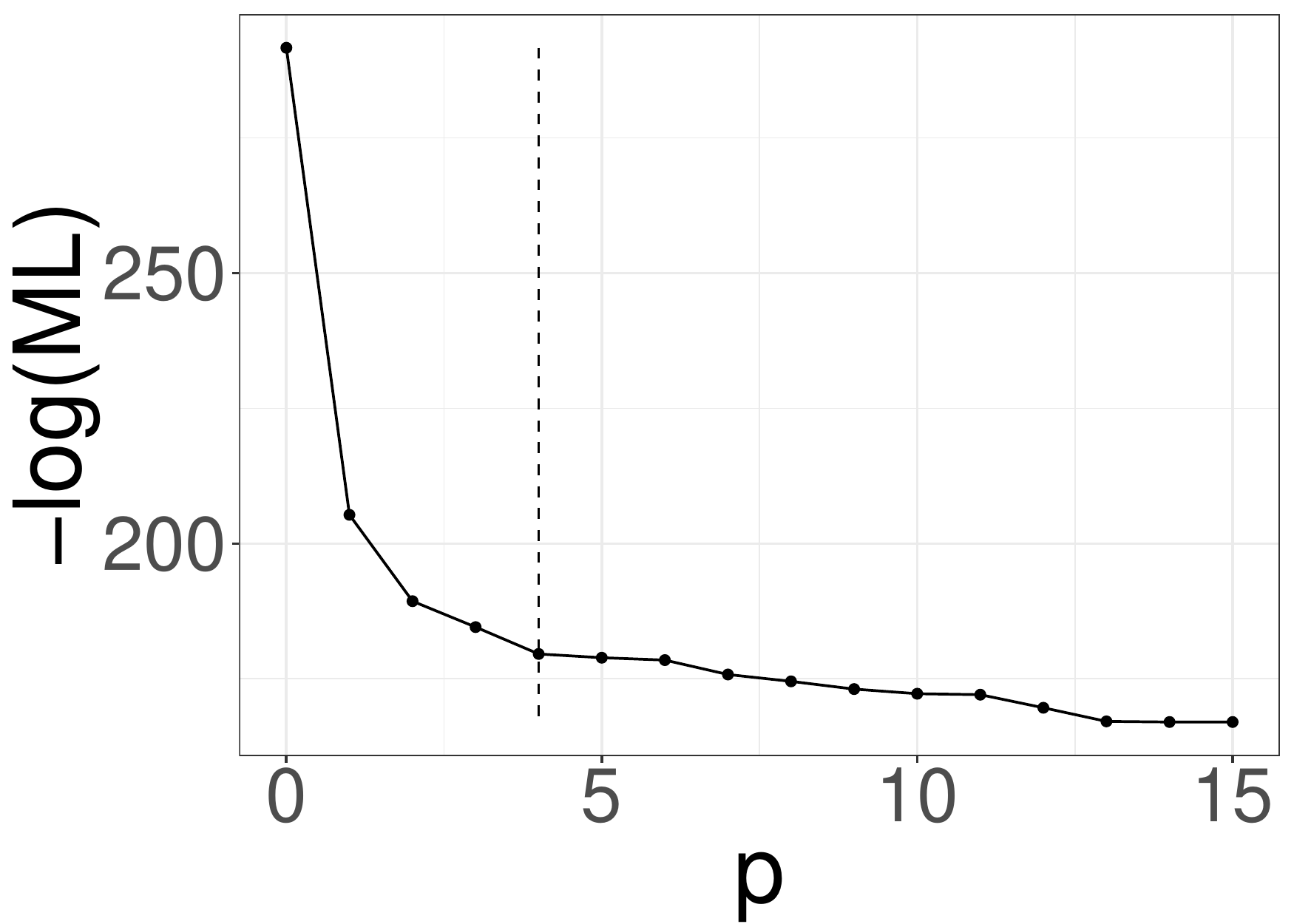}
  }
  \subfigure[]{
    \centering
 	\includegraphics[width=.22\textwidth]{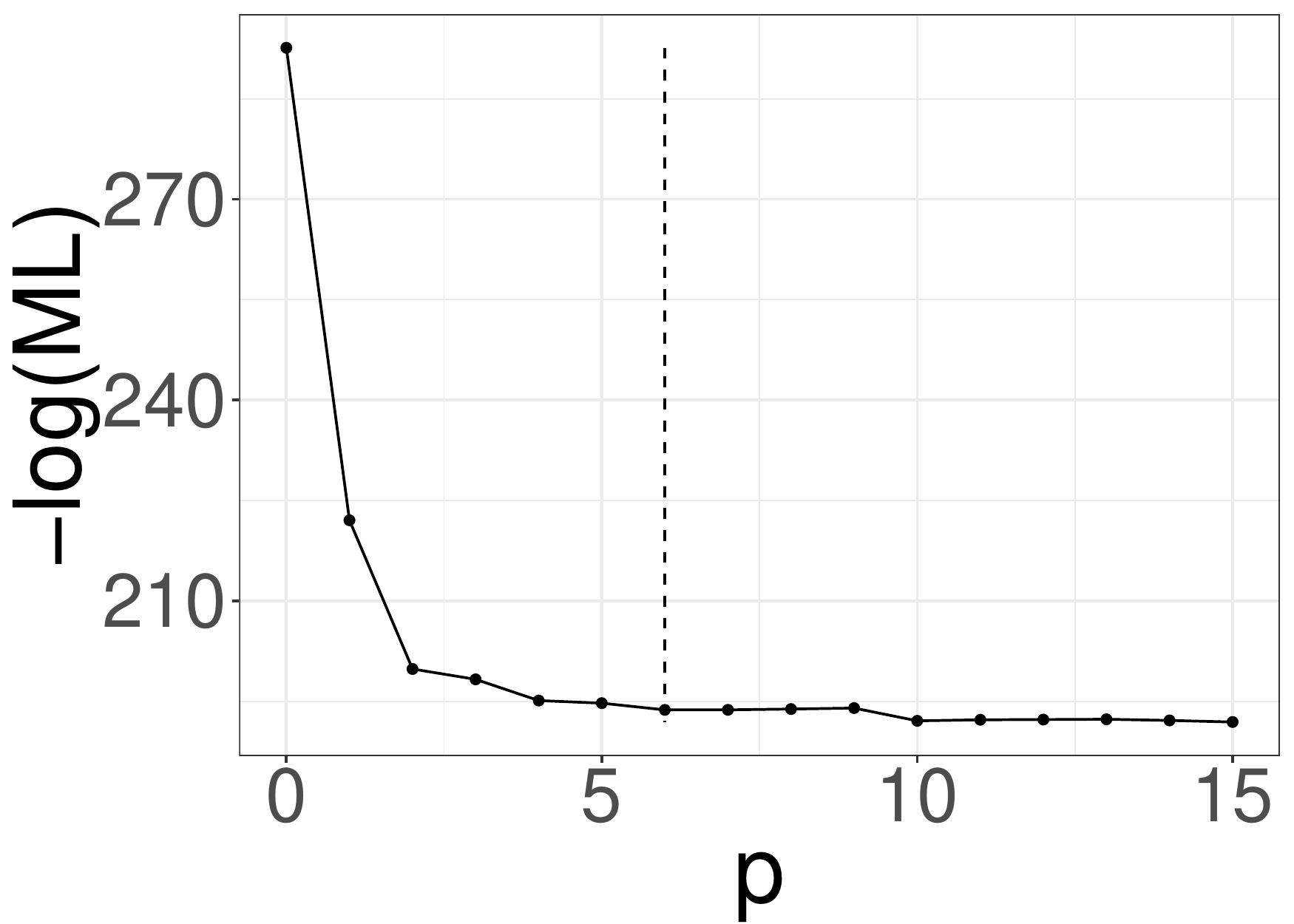}
  }
  \subfigure[]{
    \centering
 	\includegraphics[width=.22\textwidth]{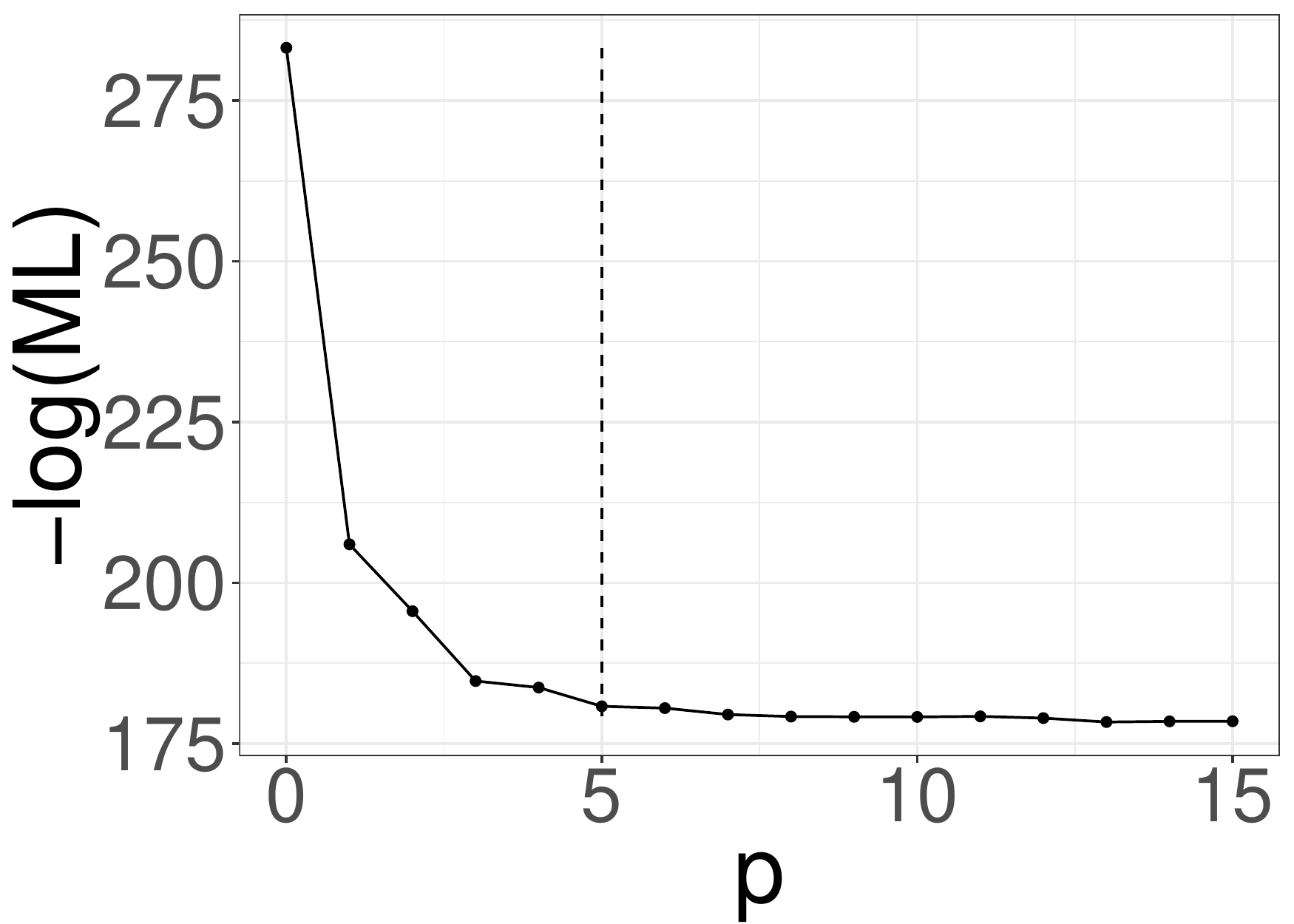}
  }
  \subfigure[]{
    \centering
 	\includegraphics[width=.22\textwidth]{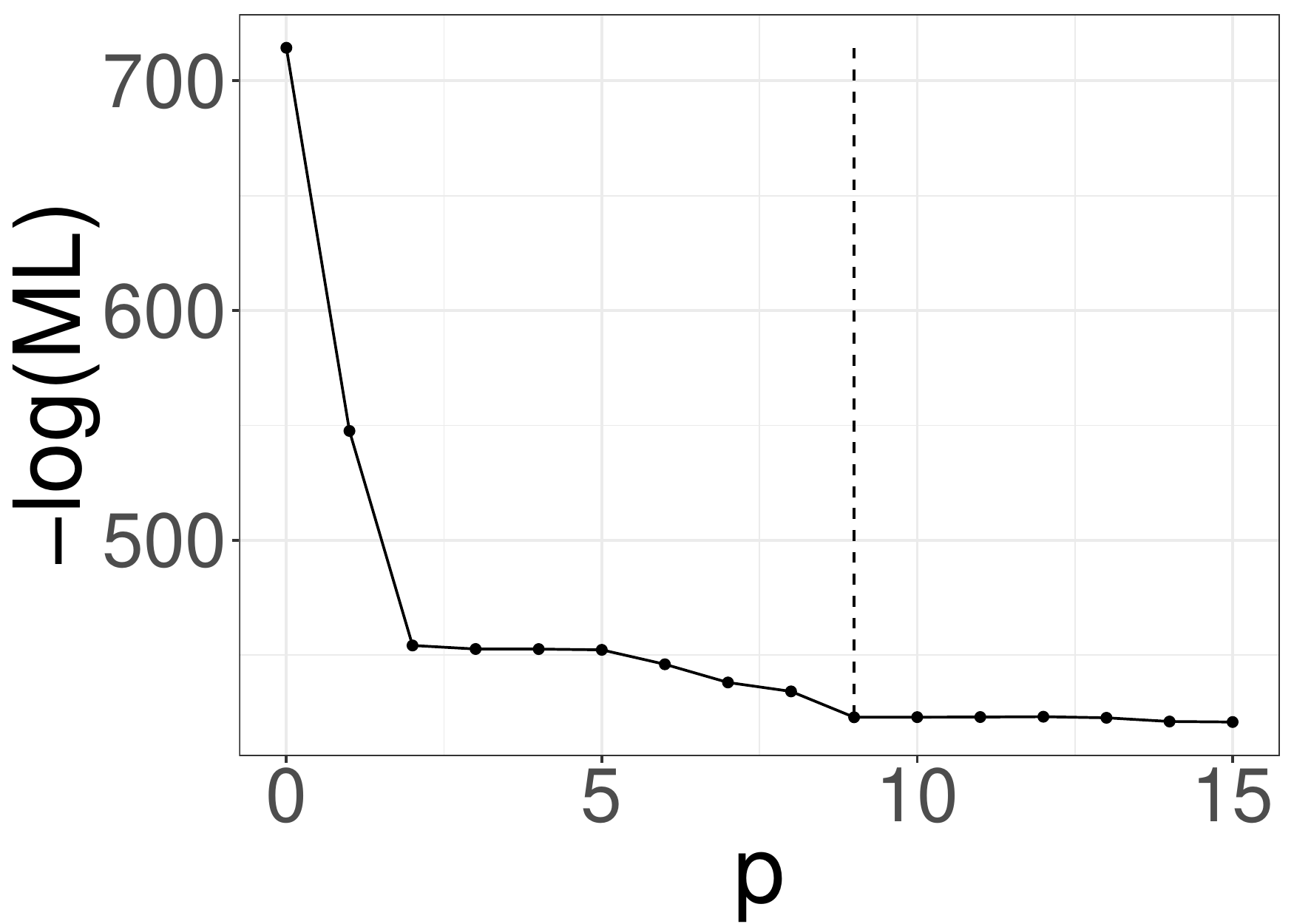}
  }
  \caption{Negative maximum log likelihood for different AR($p$)
    models applied to (a)-(c): three realizations of the ARMA(1,1)
    model with~$a=0.75, b=0.8$ of length~$n=128$ and (d) the sunspot
    data.  The respective DIC-minimizing order is visualized by a black
    dashed line.} \label{fig:visualInspectionARMA}
\end{figure}

\begin{remark}\label{rem_visualInspection}
  {\cb Under relatively weak conditions (see
    e.g.~\cite{kreiss2011range}) a linear process can be written as an
    AR($\infty$)-process with white noise errors (similarly to the
    famous Wold representation). Consequently, an AR($p$)-model with
    sufficiently large order captures the structure of a (Gaussian)
    linear process to any degree of accuracy. In this sense, the use
    of an AR-model with a sufficiently large order can be viewed as a
    nonparametric procedure, a fact, that has been exploited by
    AR-sieve-bootstrap methods for quite some time. For a recent
    mathematical analysis of the validity and limitations of this
    approach we refer to~\cite{kreiss2011range}.
	
    Consequently, an AR-model can still be used for spectral density
    estimation under misspecification as long as the order is
    sufficiently large. In this sense standard order selection
    techniques such as DIC-minimization tend to choose large orders in
    this situation. However, looking at scree-like plots of the
    negative maximum log likelihood for increasing orders one can
    often see a clear bend (elbow) in the curve (with a slow decay
    from that point on that is not slow enough to be captured by
    standard penalization techniques). Similar to the use of scree
    plots in the context of PCA, that point can be seen as a
    reasonable truncation point ('elbow criterion') where those
    features best explained by the parametric model have been
    captured. While this small model does not yet fully explain the
    data, adding more parameters is not helping the nonparametrically
    corrected procedure that we propose. In other words, we are not
    interested in an elaborate AR($\infty$) approximation but rather
    in a proxy model that captures the main parametric features of the
    data.

In the context of an autoregressive working model, we approximate the
negative maximum log-likelihood by the negative log-likelihood evaluated at the Yule-Walker
estimate. This is to ensure numerical stability and computational speed, especially for large orders.
The approximation is motivated by the asymptotic equivalence of both estimates,
see e.g. Chapter 8 in~\cite{brockwell2009time}. 
The estimate is referred to as \emph{negative maximum log-likelihood} in the text.

	Figure~\ref{fig:visualInspectionARMA} shows the scree-like
        plots for three exemplary ARMA(1,1) realizations from the
        above model as well as the sunspot data set. In all three
        realizations the elbow is clearly at $p=1$ (which is
        consistent with the AR-part of the model), while the
        DIC-criterion choses orders between 4 and 6.
	
	Table~2
        shows the simulation results for the ARMA(1,1) model and a
        fixed order of $p=1$.  
While a parametric AR($1$) model is clearly not able to explain the data (see e.g.\ the zero coverage of the uniform credibility intervals), this choice of the order significantly improves the results of the NPC procedure for the ARMA(1,1) data. In fact, the latter is now better than both the AR procedure as well as the NP procedure while at the same time the confidence in the model as indicated by $\hat\eta$ increases.

\begin{table}
\caption{Average Integrated Absolute Error (aIAE), Uniform 0.9-Credible Interval coverage (cUCI) and
average posterior model confidence~$\hat \eta$ for ARMA(1,1) data
and fixed order~$p=1$.}\label{tab:hi}
\begin{adjustbox}{max width=\linewidth}
\fbox{
\begin{tabular}{l|rrrlrrrlrrr}
& \multicolumn{11}{c}{ARMA($1,1$): $a=0.75$, $b=0.8$} \\
\hline
&&aIAE &&&& cUCI &&&& $\hat \eta$ &\\
\cline{2-4}\cline{6-8}\cline{10-12}
&    $n=64$  &  $n=128$ & $n=256$  && $n=64$  &  $n=128$ & $n=256$ && $n=64$  &  $n=128$ & $n=256$ \\
$\operatorname{AR_{(p=1)}}$&       1.179  &  1.243  & 1.289 && 0  &  0  & 0 && - & - & -  \\
$\operatorname{NPC_{(p=1)}}$&      0.942  &  0.741  & 0.586 && 0.986  &  0.969  & 0.954 && 0.601 & 0.635 & 0.670 \\
\end{tabular}}
\end{adjustbox}
\end{table}

        For the sunspot data that effect can also be seen clearly as
        the above procedure proposes to use $p=2$ in the nonparametric
        procedure while the DIC-criterion suggests $p=9$. For a
        detailed discussion of this data analysis we refer to
        Section~\ref{sec:sunspot}, similar effects for the LIGO data
        are discussed in Section~\ref{sec:ligo}.  }

\end{remark}

\subsection{Sunspot data}\label{sec:sunspot}
In this section, we analyse the yearly sunspot data from 1700 until
1987.  We take the mean-centered version of the square root of the 288
observations as input data.  We compare the AR and the NPC procedure
for fixed values~$p=1,2,3,9$. While~$p=9$ minimises the DIC, {\cb
  $p=2$ captures the main AR-features of the data as indicated by the
  elbow criterion (see Remark~\ref{rem_visualInspection} and
  Figure~\ref{fig:visualInspectionARMA}~(d)).}  The results are shown
in Figure~\ref{fig_sunspot}.
\begin{figure}[t]
 \centering
 	\includegraphics[width=\textwidth]{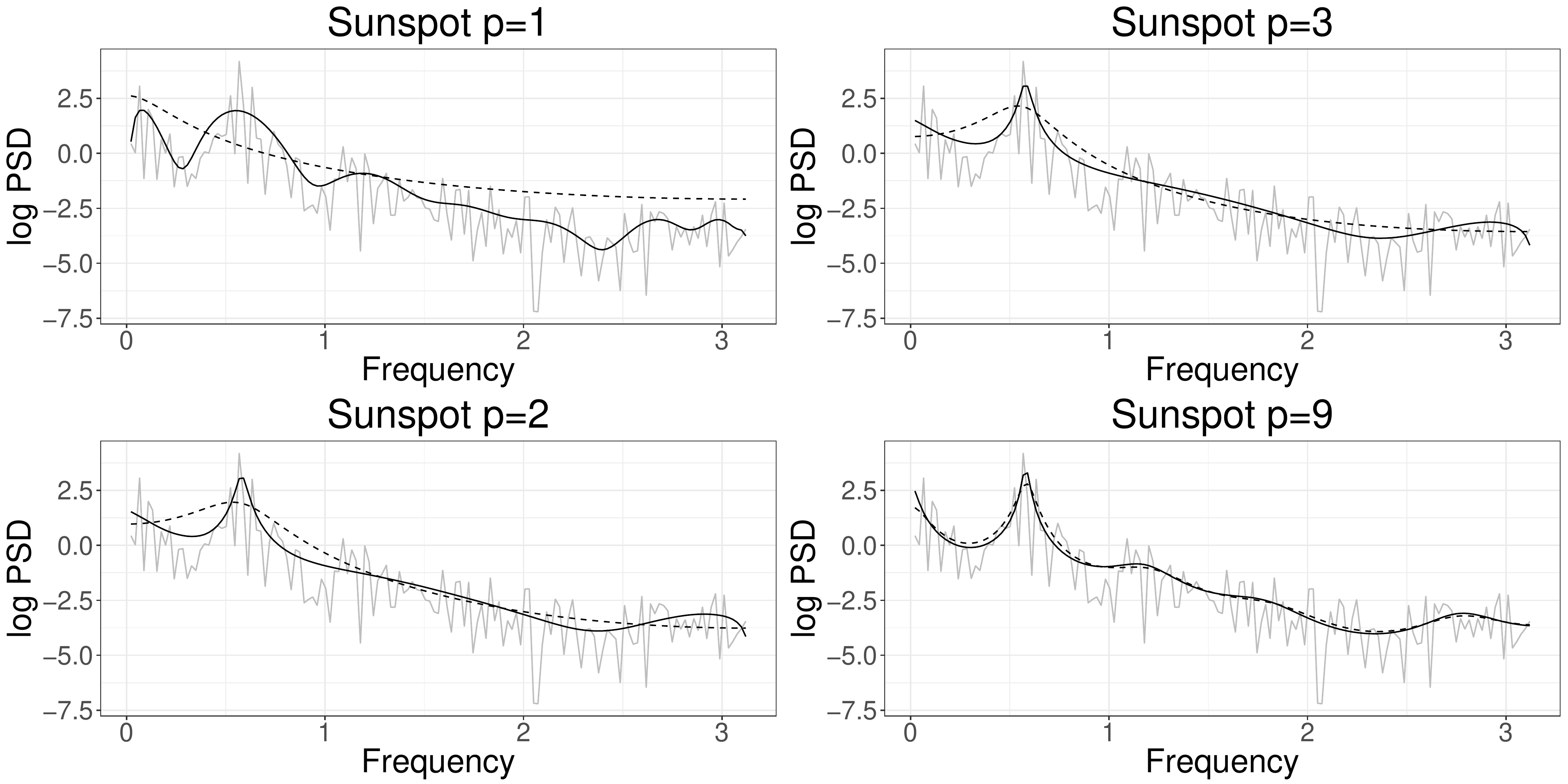}
 \caption{Posterior median spectral density estimates NPC (solid black)
   and AR (dashed black)
   for the transformed sunspot data on a logarithmic scale, for
   different autoregressive orders~$p$. The log-periodogram is
   visualised in grey.} \label{fig_sunspot}
\end{figure} {\cb While for $p=1$ the Bernstein polynomials of the
  nonparametric correction cannot yet capture the peaks sufficiently
  well, this is clearly the case for $p=2$ (the choice obtained from
  the elbow criterion): 
  While the parametric model itself can clearly not yet explain the
  data well, enough features are captured to improve the nonparametric
  correction.  For larger order choices, the estimate from the NPC
  method does not change much anymore, so that the correction does
  indeed not profit from additional parameters in the AR-model. In
  fact for $p=9$ (as indicated by DIC),} the Bayes estimator of AR and
NPC are very similar.


\pagebreak

\subsection{Gravitational wave data}\label{sec:ligo}

Gravitational waves, ripples in the fabric of spacetime caused by
accelerating massive objects, were predicted by Albert Einstein in
1916 as a consequence of his general theory of relativity, {\cb see}
\citet{einstein:1916}.  Gravitational waves originate from
non-spherical acceleration of mass-energy distributions, such as
binary inspiraling black holes, pulsars, and core collapse supernovae,
propagating outwards from the source at the speed of light.  However,
they are very small (a thousand times smaller than the diameter of a
proton) so that their measurement has provided decades of enormous
engineering challenges.

On Sept.\ 14, 2015, the Laser Interferometric Gravitational Wave
Observatory (LIGO), see \citet{ligo:2015}, made the first direct detection
of a gravitational wave signal, GW150914, originating from a binary
black hole merger~\citep{detection:2016}.  The two L-shaped LIGO
instruments (in Hanford, Washington and Livingston, Louisiana) each
consist of two perpendicular arms, each 4 kilometers long. A passing
gravitational wave will alternately stretch one arm and squeeze the
other, generating an interference pattern which is measured by
photo-detectors.  The detector output is a time series that consists
of the time-varying dimensionless strain $h(t)$, the relative change
in spacing between two test masses. The strain can be modelled as a
deterministic gravitational wave signal $s(t,\btheta)$ depending on a
vector $\btheta$ of unknown waveform parameters plus additive noise
$n(t)$, such that
\[
  h(t)= s(t,\btheta) + n(t), \quad t=1,\ldots,T.
\]

There are a variety of noise sources at the LIGO detectors.  This
includes {\em seismic} noise, due to the motion of the mirrors from
ground vibrations, earthquakes, wind, ocean waves, and vehicle
traffic, {\em thermal} noise, from the microscopic fluctuations of the
individual atoms in the mirrors and their suspensions, and {\em shot}
noise, due to the discrete nature of light and the statistical
uncertainty from the ``photon counting'' that is performed by the
photo-detectors.  In particular, LIGO noise includes high power,
narrow band, spectral lines, visible as sharp peaks in the
log-periodogram. As the LIGO spectrum is time-varying and subject to
short-duration large-amplitude transient noise events, so-called
``glitches'', a precise and realistic modelling and estimation of the
noise component jointly with the signal is important for an accurate
inference of the signal parameters $\btheta$. The current approach,
which was also used for estimating the parameters of GW150914
in \citet{estimation:2016}, is to first use the Welch method
\citep{welch:1967} to estimate the spectral density from a separate
stretch of data, close to but not including the signal and then to
assume stationary Gaussian noise with this known spectral density in
order to estimate the signal parameters.

Several approaches have been suggested in the recent gravitational
wave literature to simultaneously estimate the noise spectral density
and signal parameters. These include generalising the Whittle
likelihood to a Student-t likelihood as in \citet{roever:2011a}, similarly
modifying the likelihood to include additional scale parameters and
then marginalising over the uncertainty in the PSD
as in \citet{littenberg:2013}, using cubic splines for smoothly varying
broad-band noise and Lorentzians for narrow-band line features
as in \citet{littenberg:2015}, a Morlet-Gabor continuous wavelet basis for
both gravitational wave burst signals and glitches
as in \citet{cornish:2015}, the nonparametric approach of
\citet{Choudhuri04} using a Dirichlet-Bernstein prior
\citep{edwards:2015} and a generalisation of this using a B-spline
prior, see \citet{edwards:2017}.


We consider 1 s of real LIGO data collected during the sixth science
run (S6), recoloured to match the target noise sensitivity of Advanced
LIGO \citep{christensen:2010}. The data is differenced and then
multiplied by a Hann window to mitigate spectral leakage.  A low-pass
Butterworth filter (of order 20 and attenuation 0.25) is then applied
before downsampling from a LIGO sampling rate of 16384 Hz to 4096 Hz,
reducing the volume of data.

\begin{figure}[tb]
 \centering
 \includegraphics[width=\textwidth]{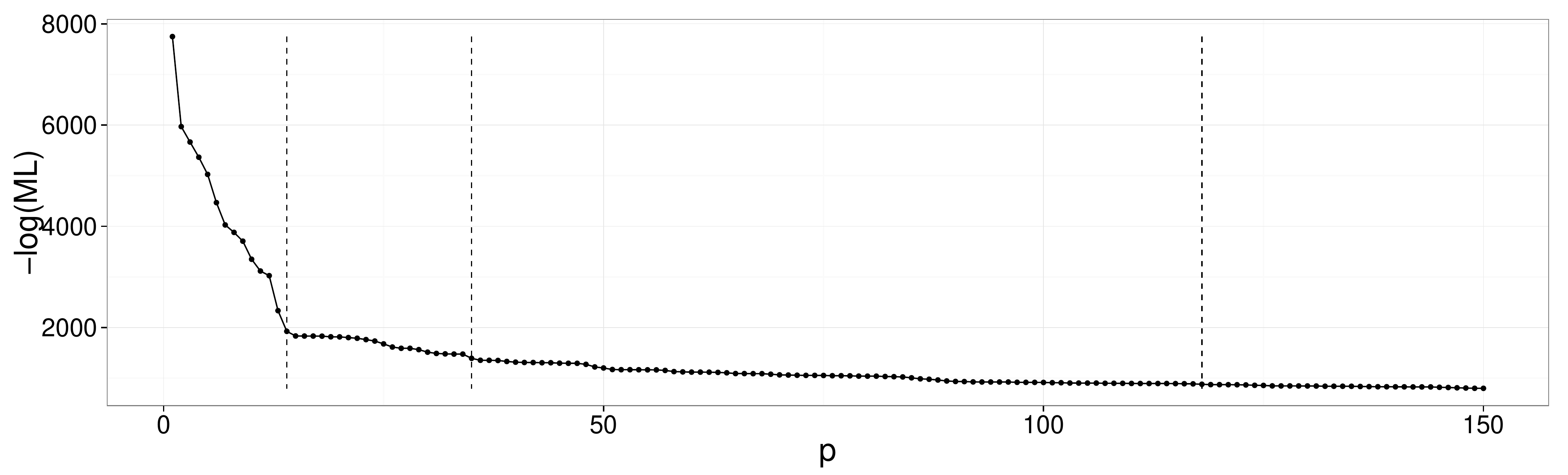}
 \caption{Negative maximum log likelihood for different AR($p$) models
   applied to Advanced LIGO S6 data.} \label{fig_screeLIGO}
\end{figure}

We first run a pure nonparametric model, corresponding to a
nonparametrically corrected likelihood with an AR(0) working model
(i.e.\ the Whittle likelihood) to estimate the spectral density.  We
then compare this to a nonparametrically corrected model with an order
of $p = 14$, where a clear elbow can be seen in the negative
log-likelihood plot (see Figure~\ref{fig_screeLIGO} and
Remark~\ref{rem_visualInspection}).  We run these simulations for
100,000 MCMC iterations, with a burn-in of 50,000, and thinning factor
of 5. Results are illustrated in Figure~\ref{fig_GW_NPC}~(a).

Even though $k$ converged to $k \approx 900$ mixture components, it is
clear that the Bernstein-Dirichlet prior together with the Whittle
likelihood is not flexible enough to estimate the sharp peaks of the
LIGO spectral density. The parametric AR(14) model (estimated using
the Bayesian autoregressive sampler described in Appendix~\ref{sec_appendixSampling}) captures the four main peaks but not their
sharpness.  Additionally, it does not capture the structure well in
the frequency bands~0 to~450~Hz as well as larger than~1100~Hz.  When
compared to the AR(0) model, the nonparametrically corrected model
based on $p=14$ estimates the sharp peaks much better.  Furthermore,
it sharpens all four peaks of the AR(14)-model (with a slight
exception around 400~Hz, where seemingly two very sharp peaks overlap,
a feature that is not captured by the AR(14) model at all).  In the
frequency bands 0 to 450~Hz as well as larger than 1100~Hz, where the
parametric model fails altogether, the correction yields similar
results to the nonparametric Whittle procedure.  Similarly to the
nonparametric Whittle procedure $k$ tends towards $k_{\max}=800$
indicating that the Bernstein-Dirichlet prior together with an
AR($14$)-model is not yet flexible enough for this data set.

\begin{figure}[tb]
 \subfigure[]{
   \centering
   \includegraphics[width=\textwidth]{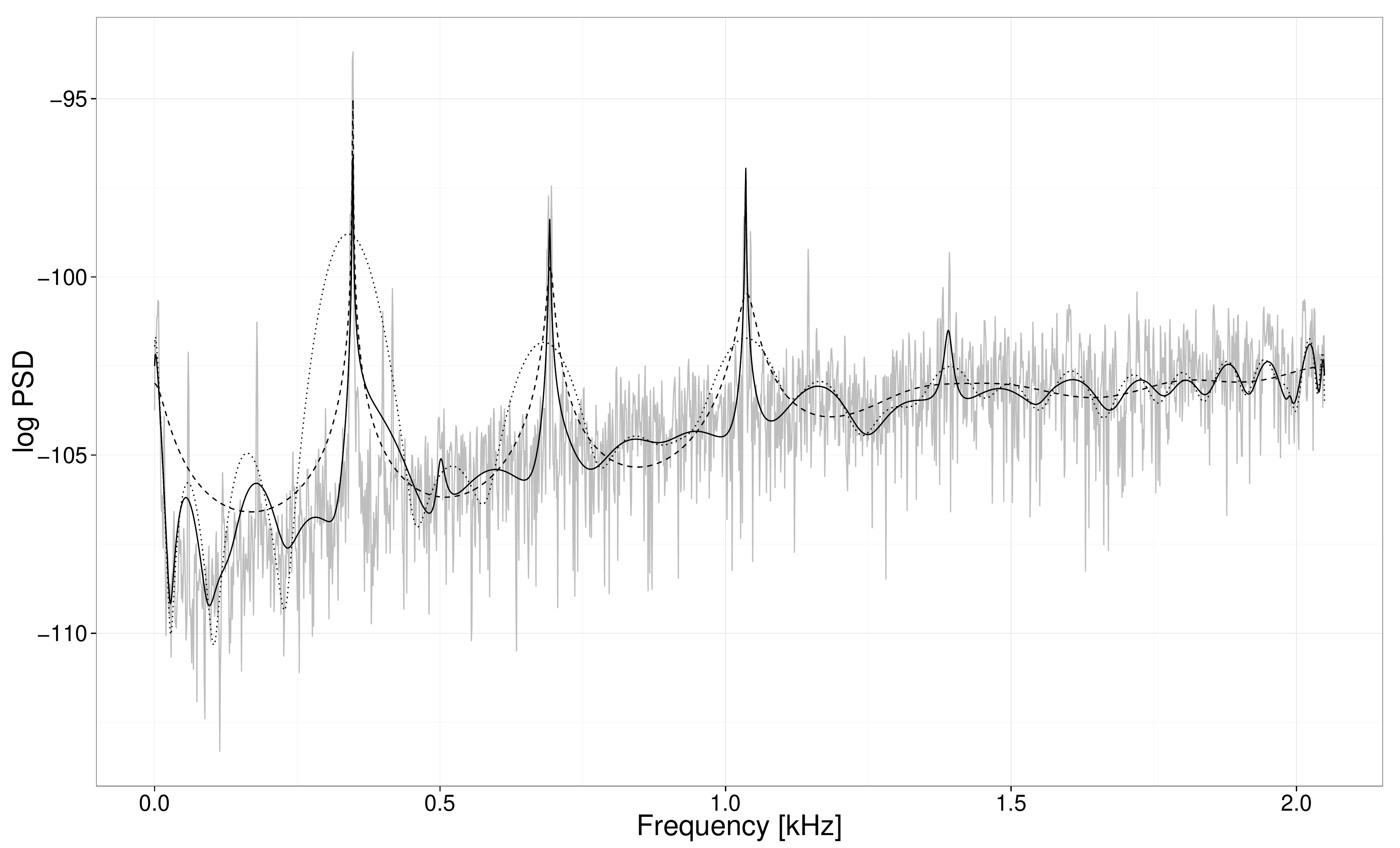}
 }
 \subfigure[]{
   \centering
   \includegraphics[width=\textwidth]{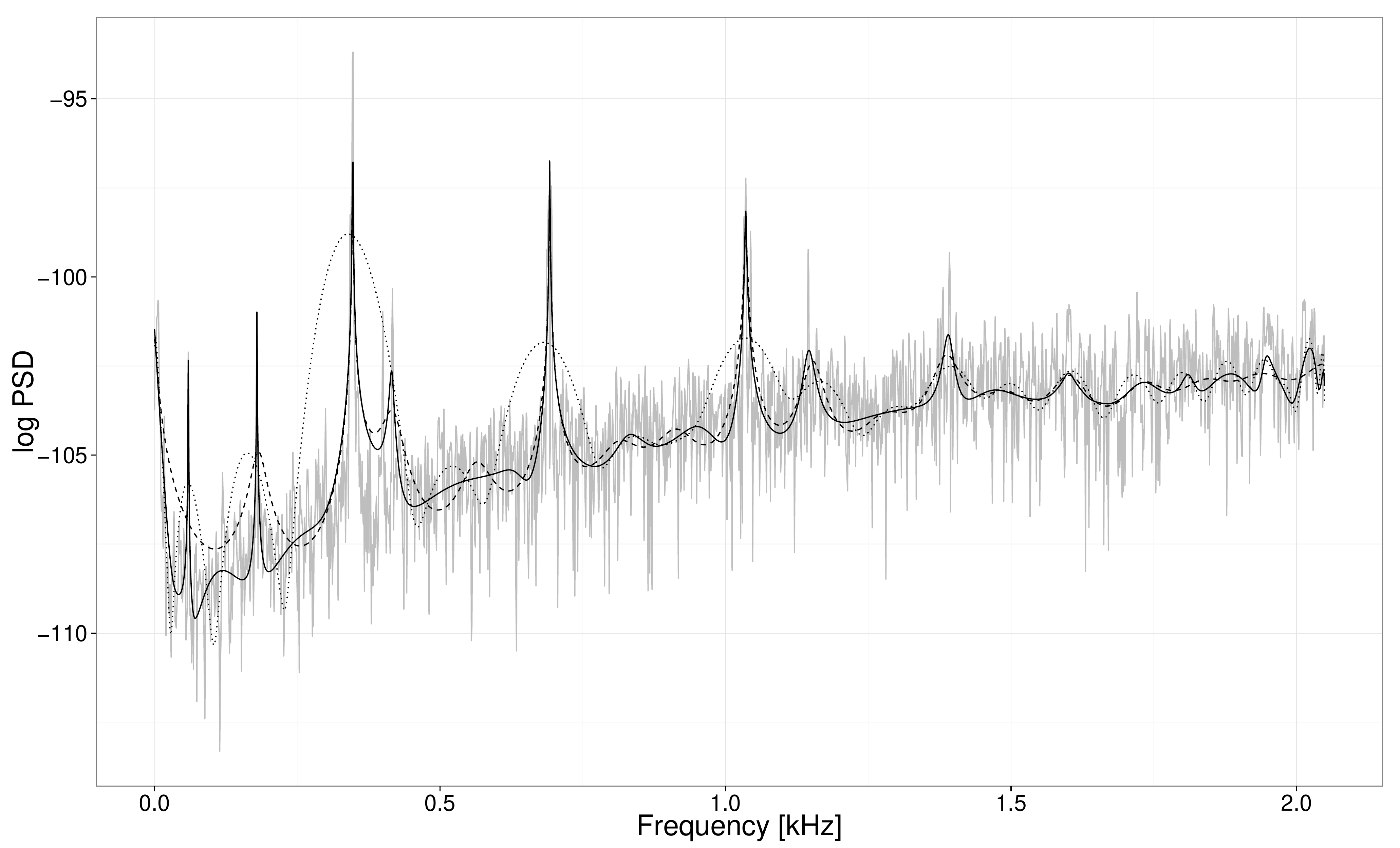}
 }

 \caption{Estimated log spectral density for a 1 s segment of LIGO S6
   data. The posterior median log spectral density estimate of NPC
   under an AR($p$) working model (solid black), AR($p$) (dashed
   black), and NPC under an AR(0) working model (dotted black) are
   overlaid with the log periodogram (grey), where (a) depicts
   $p=14$ and (b) depicts $p=35$.} \label{fig_GW_NPC}
\end{figure}

\afterpage{\clearpage}

Looking closer at Figure~\ref{fig_screeLIGO}, the negative
log-likelihood between the models AR(14) and AR(35) decreases by 533,
which is not as sharp as the elbow at $p=14$, but still significant --
keeping in mind that the LIGO data is a very complex data set -- much
more so than the ARMA(1,1) or the sunspot data.  There is a moderately
sized jump before $p=35$, while afterwards the descend slows down
significantly. In fact the BIC chooses an order of $p=118$, where the
log-likelihood reaches the level of 879, showing that the difference
between $p=14$ and $p=35$ (of 533) is comparable to the one between
$p=35$ and $p=128$ (of 512). This indicates that there is indeed
another change of gradient around $p=35$. When looking at penalized
likelihoods, this is also the point, where different penalizations
start to obviously diverge.  The results for $p=35$ can be found in
Figure~\ref{fig_GW_NPC}~(b).

The parametric AR($35$) model already provides a reasonable fit to the
periodogram, picking up the major peaks (with the exception around
100~Hz), but under- and overestimates some of the peaks. In particular
there are still major problems in the frequency bands 0 to 300 Hz and
400 to 700~Hz.  The NPC procedure with $p=35$ keeps the peaks that
have been capture well by the parametric model but corrects problems
most prominently in the above mentioned frequency bands. It is worth
mentioning that the correction works in several ways: Sharpening
existing peaks (e.g.\ at 0 Hz), adding new peaks (e.g.\ at 100 Hz) as
well as smoothing out some erroneous peaks (e.g.\ at 600 Hz).
Overall, the resulting estimate seems to capture the structure quite
well in all frequency bands.  This impression is complemented by the
results of the NPC method with AR($35$) working model together with
the pointwise and uniform credible bounds obtained from the procedure
in Figure~\ref{fig_GW_CI}.

\begin{figure}[tb]
 \centering
 \includegraphics[width=\textwidth]{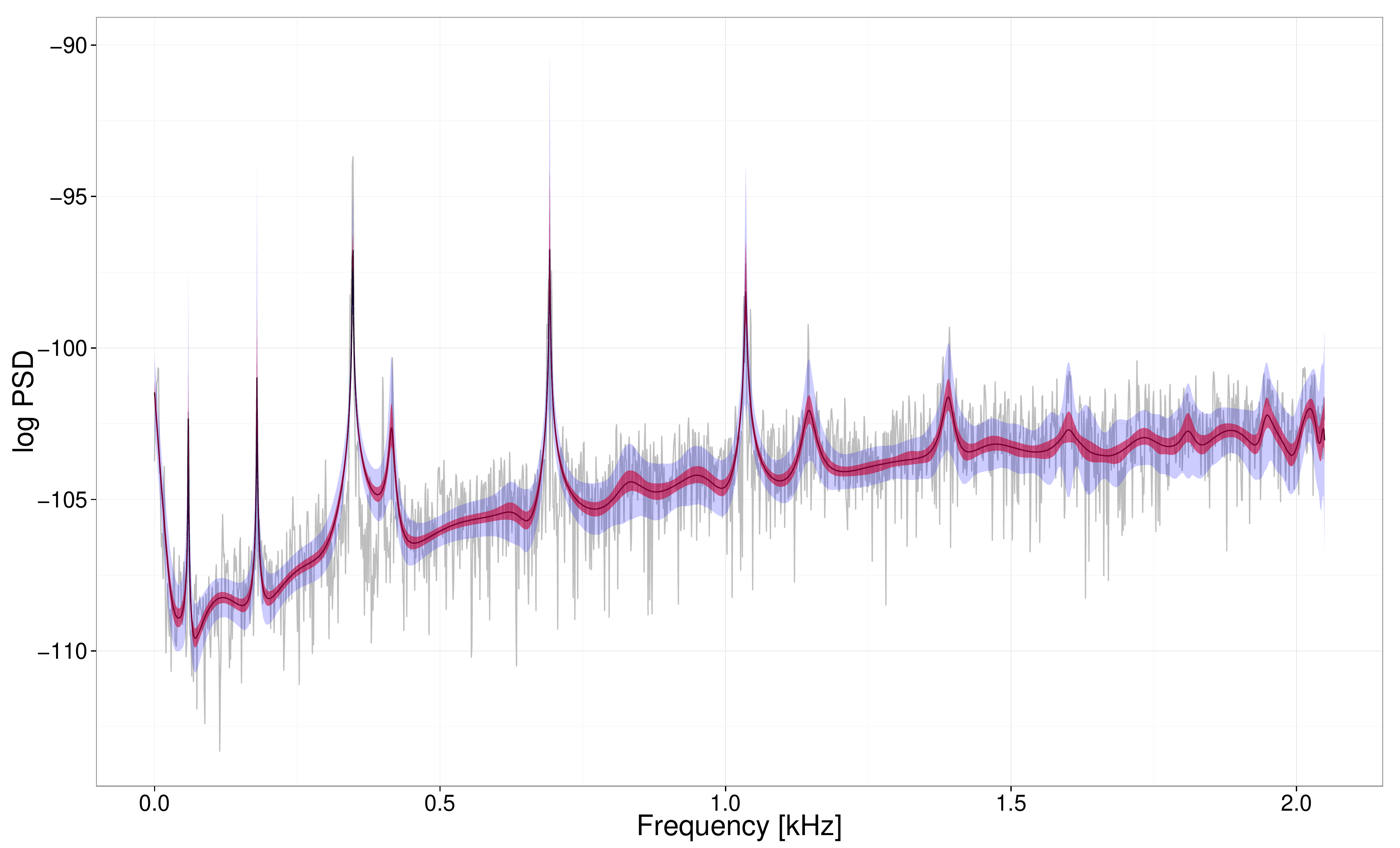}
 \caption{Estimated log spectral density for a 1 s segment of LIGO S6
   data. The posterior median log spectral density estimate of NPC
   under an AR(35) working model (solid black), pointwise 0.9-credible
   region (shaded red), and uniform 0.9-credible region (shaded violet)
   are overlaid with the log periodogram (grey).} \label{fig_GW_CI}
\end{figure}

\afterpage{\clearpage}

\section{Conclusion}\label{sec:summary}
In this work we propose a nonparametric correction of a parametric
likelihood to obtain an approximation of the true (inaccessible)
likelihood of a stationary time series. This approach extends the
famous approximation by Whittle.  For Gaussian data, the Whittle
likelihood, the nonparametric correction as well as the true
likelihood are asymptotically equivalent.  Secondly, we propose a
Bayesian procedure for spectral density estimation, where a parametric
likelihood is used with a Bayesian nonparametric spectral correction.
We show consistency of the resulting pseudo-posterior distribution for
a fixed parametric likelihood.  Furthermore, we present a Bayesian
semiparametric procedure that combines inference about the parametric
working model with the nonparametric correction.  The extent of the
contribution of the parametric spectral density to the spectral
density estimate is controlled by a shape confidence parameter.
Simulation results have shown that this procedure inherits the
benefits from the parametric working model if the latter is
well-specified {\cb or describe a part of the features of the data
  well}, while in the misspecified case the results are comparable to
the usage of the Whittle likelihood.

Regarding future work, it is interesting to investigate whether in the
non-Gaussian case the class of time series for which asymptotically
consistent inference holds can be enlarged by choosing an appropriate
model.  It is important to understand the distributional influence in
the non-Gaussian case both in finite samples and asymptotically.  As
an example, in a bootstrap context, it suffices to capture the fourth
order structure of a linear model to obtain asymptotically valid
second-order frequentist properties of the autocovariance structure
\citet{Dahlhaus96, Kreiss03}.  As suggested by
\citet{kleijn2012bernstein} in a parametric setting, this property
does not simply carry over to a Bayesian context.  Preliminary results
for non-Gaussian autoregressive time series however have shown
considerable benefits (with respect to first and second order
frequentist properties) when the innovation distribution is
well-specified in comparison to a Gaussian model.  Since any
parametric likelihood is susceptible to misspecification, the ultimate
goal is to consider a Bayesian nonparametric model for the innovation
distribution, such as Dirichlet mixtures of normals.
Further directions for future work are automation of the {\cb elbow criterion} 
discussed in Remark~\ref{rem_visualInspection}.
Instead of choosing a fixed order in advance, an automation might as well 
serve as a guideline for specifying a prior on the AR order parameter, which
can be included in the inference by means of RJMCMC (c.f. Remark~\ref{rem_arOrder}).

\section*{Acknowledgements}
This work was supported by DFG grant KI 1443/3-1. Furthermore, the
research was initiated during a visit of the {\cb fourth} author at
Karlsruhe Institute of Technology (KIT), which was financed by the
German academic exchange service (DAAD).  Some of the preliminary
research was conducted while the first author was at KIT, where her
position was financed by the Stifterverband f\"ur die deutsche
Wissenschaft by funds of the Claussen-Simon-trust.  We also thank the
New Zealand eScience Infrastructure (NeSI) and the
Universit\"atsrechenzentrum (URZ) Magdeburg for their high performance
computing facilities, and the Centre for eResearch at the University
of Auckland and J\"org Schulenburg for their technical support.

\bibliographystyle{plainnat}

\appendix
\section{Appendix: Proofs}\label{sec:proofs}

\begin{proof}[of \eqref{eq_whittle}]
Since $F_n$ is orthonormal it holds $|\det F_n|=1$ and hence by an application of the transformation theorem
\begin{eqnarray*}
	p^W_{\mathbf{Z}_n=F_n^T\mathbf{\tilde{Z}}_n}(\mathbf{z}_n|f)  &=& |\det F_n|\; p_{\mathbf{\tilde{Z}}_n}^W(F_n\mathbf{z}_n|f)  \propto \det(D_n)^{-1/2}\exp\left(-\frac 1 2 \boldsymbol{\tilde z}_n^TD_n^{-1}\boldsymbol{\tilde z}_n\right)\\
&\propto& \exp\left(-\frac 1 2 \sum_{j=0}^{n-1}\log f(\lambda_j) -\frac 1 2  \boldsymbol{\tilde z}_n^TD_n^{-1}\boldsymbol{\tilde z}_n\right),
\end{eqnarray*}
where $\boldsymbol{\tilde z}_n=F_n \boldsymbol z_n$.
Finally, by \eqref{eq_perio_four}, $I_n(\lambda_k)=I_n(\lambda_{n-k})$ as well as $f(\lambda_k)=f(\lambda_{n-k})$ it holds
\begin{align*}
	& \boldsymbol{\tilde{z}}_n^TD_n^{-1}\boldsymbol{\tilde{z}}_n=\frac{\tilde{z}_n^2(0)}{2\pi f(0)}+\frac{\tilde{z}_n^2(n/2)}{2\pi f(\pi)}1_{\{n\text{ even}\}}+\sum_{j=1}^N\frac{\tilde{z}^2_n(2j-1)+\tilde{z}^2_n(2j)}{2\pi f(\lambda_j)}\\ 
	&=\frac{I_{n,\lambda_0}(\boldsymbol z)}{f(0)}+\frac{I_{n,\lambda_{n/2}}(\boldsymbol z)}{f(\pi)}1_{\{n\text{ even}\}}+\sum_{j=1}^N\frac{2 I_{n,\lambda_j}(\boldsymbol z)}{f(\lambda_j)}=\sum_{j=0}^{n-1}\frac{I_{n,\lambda_j}(\boldsymbol z)}{f(\lambda_j)},
\end{align*}
yielding the assertion.
\end{proof}
\par
\begin{proof}[of Remark~\ref{rem_identifiability}]
  The spectral density corresponding to a Gaussian white noise with
  variance $\sigma^2$ is given by $\sigma^2/2\pi$, hence
  $\sigma^2\,C_n= D_n$ and
	\begin{align*}
		p_{\text{i.i.d. }N(0,\sigma^2)}^C(\bZ_n|f)&\propto \frac{\det C_n^{-1/2}}{\sigma^n}\exp\left(-\frac{1}{ 2\sigma^2} \mathbf{Z}_n^TF_n^TC_n^{-1/2}F_nF_n^TC_n^{-1/2}F_n\mathbf{Z}_n\right)\\
		&\propto \det D_n^{-1/2}\exp\left(-\frac 1 2 (F_n\mathbf{Z}_n)^TD_n^{-1}(F_n\mathbf{Z}_n)\right),
	\end{align*}
which can be shown to be  the Whittle likelihood  analogously to the proof of equation~\ref{eq_whittle} since $F_n \mathbf{Z}_n=\boldsymbol{\tilde{Z}}_n$.
\end{proof}

\begin{proof}[of Proposition~\ref{cor_31}]
  If $f=f_{\text{param}}$ then $C_n=\text{Id}$ and
  $F_n^TC_n^{-1/2}F_n=\text{Id}$, hence a) follows.  For b), consider
  a time series distributed according to the corrected likelihood,
  i.e.\ $\mathbf{X}_n=(X_1,\ldots,X_n)\sim p_{\text{param}}^C$. An
  application of the transformation theorem shows that $\mathbf
  Y_n=F_n^TC_n^{-1/2}F_n\mathbf X_n\sim p_{\text{param}}$ on noting
  that $(F_n^TC_n^{-1/2}F_{n})^{-1}=F_n^TC_n^{1/2}F_n$ and
  $\operatorname{det}((F_n^TC_n^{-1/2}F_{n})^{-1})=\sqrt{\mbox{det}(C_n)}$.
  By \eqref{eq_perio_four} it holds
	\begin{align*}
&	I_{n,\lambda_j}(\mathbf X_n)=I_{n,\lambda_j}(F_n^TC_n^{1/2}F_n \mathbf Y_n)
=\frac{f(\lambda_j)}{f_{\text{param}}(\lambda_j)}\, I_{n,\lambda_j}(\mathbf Y_n)
\end{align*}
as $F_n (F_n^TC_n^{1/2}F_n \mathbf Y_n)=C_n^{1/2} (F_n\mathbf
Y_n)$. By~\citet{brockwell2009time}, Proposition 10.3.1., it holds $\E
I_{n,\lambda_j}(\mathbf Y_n)=f_{\text{param}}(\lambda_j)+o(1)$ where
convergence is uniform in $j$ (recall that the time series is mean
zero). From this the assertion follows because $f_{\text{param}}$ is
bounded from below by assumption.
\end{proof}

\begin{proof}[of Theorem~\ref{theorem_contiguity}]\citet{Ghosal04} proved  mutual contiguity of the true Gaussian and the Whittle likelihood
  in the frequency domain which carries over to the time domain by an
  application of the transformation theorem because $F_n$ is
  bijective.  Hence, it is sufficient to show mutual contiguity of the
  corrected parametric likelihood and the Whittle likelihood.
  Following the proof of Theorem 1 in \citet{Ghosal04} it is enough to
  show that their log-likelihood ratio is a tight sequence under both
  $p^C$ as well as $p^W$. To this end, note that it holds
\begin{align*}
	&p^W\propto \det D_n^{-1/2}\exp\left( -\frac 12 \mathbf{Z}_n^TF_n^TD_n^{-1}F_n\mathbf{Z}_n \right),\\
	&p^C\propto\det C_n^{-1/2} \;\det \Sigma_{n,\text{param}}^{-1/2} \exp\left( -\frac 12 \mathbf{Z}_n^TF_n^TC_n^{-1/2}F_n\Sigma_{n,\text{param}}^{-1}F_n^TC_n^{-1/2}F_n\mathbf{Z}_n  \right),
\end{align*}
where $\Sigma_{n,\text{param}}$ is the covariance matrix of the
corresponding parametric time series, e.g. for the AR($p$)-case the
covariance matrix associated with likelihood~\eqref{eq_param_dens}.
Hence, the log-likelihood ratio is given by
\begin{align*}
	&\frac 1 2 \left( \log \det D_n -\log \det C_n -\log \det \Sigma_{n,\text{param}} \right)\\
	&+\frac 1 2 \left( \mathbf{Z}_n^TF_n^TD_n^{-1}F_n\mathbf{Z}_n-\mathbf{Z}_n^TF_n^TC_n^{-1/2}F_n\Sigma_{n,\text{param}}^{-1}F_n^TC_n^{-1/2}F_n\mathbf{Z}_n \right)=:\frac 1 2 A_n+\frac 1 2 B_n(\mathbf{Z}_n).
\end{align*}
Defining $D_{n,\text{param}}$ analogously to $D_n$ as in
\eqref{eq_def_Dn} with the nonparametric spectral density $f$ replaced
by the parametric version $f_{\text{param}}$ as e.g. for AR($p$) given
in \eqref{eq_def_sd_param}, we get
\begin{align*}
	&A_n=\log\det D_{n,\text{param}} -\log \det \Sigma_{n,\text{param}}=O(1),
\end{align*}
where the boundedness follows from Lemma A.1 in \citet{Ghosal04}.  To
obtain stochastic boundedness of $B_n(\mathbf{Z}_n)$ under $p^W$ as
well as $p^C$ we will show boundedness of the expectation and
variance. Following \citet{Ghosal04} we get under $p^W$ (i.e.\ $F_n
\mathbf Z_n\sim N(0,D_n)$)
\begin{align*}
	&\E B_n(\mathbf{Z}_n)=\mbox{tr}\left( I_n-C_n^{-1/2}F_n\Sigma^{-1}_{n,\text{param}}F_n^TC_n^{-1/2}D_n\right).	
\end{align*}
Because $\mbox{tr}(AB)=\mbox{tr}(BA)$, it holds by $C_n^{-1/2}D_nC_n^{-1/2}=D_{n,\text{param}}$
\begin{align}
	&\mbox{tr}(C_n^{-1/2}F_n\Sigma^{-1}_{n,\text{param}}F_n^TC_n^{-1/2}D_n )
	=\mbox{tr}(F_n\Sigma^{-1}_{n,\text{param}}F_n^TD_{n,\text{param}}),\label{eq_proof_th31_1}
\end{align}
the assertion follows by Lemma A.2 in \citet{Ghosal04} by the linearity
of the trace.  Similar arguments yield the assertion under $p^C$
(i.e.\ $\mathbf Z_n\sim
N(0,F_n^TC_n^{1/2}F_n\Sigma_{n,\text{param}}F_n^TC_n^{1/2}F_n))$
noting that
\begin{align*}
	&\E B_n(\mathbf{Z}_n)=\mbox{tr}\left( F_n^TD_n^{-1}C_n^{1/2}F_n\Sigma_{n,\text{param}}F_n^TC_n^{1/2}F_n -I_n\right)
	=\mbox{tr}\left( D_{n,\text{param}}^{-1}F_n\Sigma_{n,\text{param}}F_n^T-I_n \right).	
\end{align*}
Concerning the variance we get under $p^W$
\begin{align*}
	&\var B_n(\mathbf{Z}_n)=2\,\mbox{tr}\left(\left( I_n-C_n^{-1/2}F_n\Sigma^{-1}_{n,\text{param}}F_n^TC_n^{-1/2}D_n \right)^2\right).	
\end{align*}
Similar arguments as above yield
\begin{align*}
	&\mbox{tr}\left( \left(C_n^{-1/2}F_n\Sigma^{-1}_{n,\text{param}}F_n^TC_n^{-1/2}D_n\right)^2  \right)	=\mbox{tr}\left( \left(F_n\Sigma_{n,\text{param}}^{-1}F_n^TD_{n,\text{param}}\right)^2 \right).
\end{align*}
Together with \eqref{eq_proof_th31_1} this yields
\begin{align*}
	&\var B_n(\mathbf{Z}_n)=2\,\mbox{tr}\left( \left( I_n- F_n\Sigma_{n,\text{param}}^{-1}F_n^TD_{n,\text{param}} \right)^2 \right)
\end{align*}
by $\mbox{tr}( (I_n-A )^2
)=\mbox{tr}(I)-2\mbox{tr}(A)+\mbox{tr}(A^2)$. Hence, the assertion
follows by Lemma A.2 in \citet{Ghosal04}. Analogous assertions yield
the result under $p^C$.
\end{proof}

\citet{Ghosal99} give sufficient conditions for the consistency of the
posterior distribution when using Bernstein polynomial priors in terms
of the existence of exponentially powerful tests for testing
$H_0:\theta=\theta_0$ and prior positivity of a Kullback-Leibler
neighbourhood. But these require i.i.d. observations. To prove
posterior consistency under the Whittle likelihood, \citet{Choudhuri04}
extend this result to independent but not identically distributed
observations and apply this to the periodogram ordinates which are
independent exponential random variables under the Whittle
likelihood. However, periodogram ordinates under the corrected
likelihood are no longer independent, therefore this theorem is not
applicable. We give an extension to non-independent random variables
in the following theorem, which is needed to prove
Theorem~\ref{theorem:consistency-corrected}:
\begin{theorem}\label{theorem:consistency-general}
  Let $\Z_n=(Z_{1,n},\ldots,Z_{n,n})$ be random vectors with
  probability distribution $P^n_{\theta}$ and corresponding pdf
  $p_n(\cdot|\theta)$. Let $\theta_0 \in \Theta$, $U_n\in {\cal T}$,
  where ${\cal T}$ denotes the Borel $\sigma$-algebra on $\Theta$, and
  $\Pi$ a probability distribution on $(\Theta,{\cal T})$. Define
\begin{eqnarray*}
K_n(\theta_0,\theta) &=& \E_{\theta_0}\left[ \log \frac{p_n(\Z_n|\theta_0)}{p_n(\Z_n|\theta)} \right] \mbox{ and }\\
V_n(\theta_0,\theta) &=& \var_{\theta_0}\left[ \log \frac{p_n(\Z_n|\theta_0)}{p_n(\Z_n|\theta)} \right] .
\end{eqnarray*}
Under the following assumptions of prior positivity of neighbourhoods
and existence of tests:
\begin{itemize}
\item[(C1)] There exists a set $B\in {\cal T}$ with $\Pi(B)>0$ such that
\begin{itemize}
\item[(a)] $\frac{1}{n^2} V_n(\theta_0,\theta) \rightarrow 0 \qquad
  \mbox{for all } \theta \in B$, and
\item[(b)] $\lim \inf_{n\rightarrow \infty} \Pi(\{\theta \in B:
  \frac{1}{n}K_n(\theta_0,\theta) < \epsilon\}) >0 \qquad \forall
  \epsilon>0$.
\end{itemize}
\item[(C2)] There exists test functions $\{\phi_n\}$, subsets
  $\Theta_n \subset \Theta$, and constants $K_1,K_2,k_1,k_2>0$ such
  that
\begin{itemize}
	\item[(a)] $\E_{\theta_0} \phi_n \rightarrow 0$.
\item[(b)]$\sup_{\theta \in U_n^c\cap \Theta_n} \E_\theta(1-\phi_n)\ls K_1 \e^{-k_1 n}$, and 
\item[(c)]$\Pi(\Theta_n^c) \ls K_2 \e^{-k_2 n}.$
\end{itemize}
\end{itemize}
Then 
\[\Pi(\theta\in U_n^c|\Z_n)\rightarrow 0 \qquad \mbox{in } P_{\theta_0}^n \mbox{-probability}.\]
\end{theorem}
\begin{proof}
  The proof is completely analogous to the proof of Theorem A.1 in
  \citet{Choudhuri04}.
\end{proof}

The following lemma replaces Lemma B.3 in  \citet{Choudhuri04}.
\begin{lemma}\label{lemma_b3}
Let $\Z_n \sim N(0, F_n^TD_n(f)^{1/2} S_n D_n(f)^{1/2}F_n)$ where
$S_n$ is a symmetric positive definite matrix and $D_n(f)=\mbox{diag}(f(\lambda_1),\ldots,f(\lambda_n))$.
With $c(\lambda_i)=f(\lambda_i)/f_{\text{param}}(\lambda_i)^\eta$ consider testing 
\begin{align*}
&H_0: c(\lambda_i)=c_{0}(\lambda_i), \quad  \mbox{ where } c_{0}(\lambda_i)\ls a \quad \mbox{ for } i=1,\ldots,n,\quad \mbox{ against }\\
&H_1: c(\lambda_i)=c_{1}(\lambda_i), \quad \mbox{ where } c_{1}(\lambda_i)<c_{0}(\lambda_i) -\epsilon \quad \mbox{ for } i=1,\ldots,n, \end{align*}
where $a>\epsilon>0$ do not depend on $n$.
Then there exists a test $\phi_n$ and constants $\beta_1,\beta_2>0$ depending only on $a$ and $\epsilon$ such that
\begin{align*}
&\E_{H_0}(\phi_n) \ls  \e^{-\beta_1 n}\qquad \mbox{ and}\qquad
\E_{H_1} (1-\phi_n) \ls \e^{-\beta_2 n}.
\end{align*}
\end{lemma}
\begin{proof}
Consider a test  $\phi_n$ that rejects $H_0$ if
\begin{align*}
	T_n=\bZ_n^TF_n^TD_n(f_0)^{-1/2}S_n^{-1}D_n(f_0)^{-1/2} F_n\bZ_n < nx
\end{align*}
with a critical value $x>0$ with $x\neq 1$ and $x\neq 1-\eps/a$. 
\par
Denote $\bY_n\sim N(0,I_n)$, then 
\begin{align*}
	T_n\ovunset{=}{\mathcal{D}}{H_0} \bY_n^T\bY_n\ovunset{=}{\mathcal{D}}{H_0} \chi^2_n.
\end{align*}
Consequently, the moment generating function of $T_n$ under $H_0$ is
given by $E\left[\e^{tT_n}\right]=(1-2t)^{-n/2}$ and exists for
$t<1/2$. By an application of the Markov inequality, we get for all
$z>0$
\begin{align*}
	&P_{c_0}\left(T_n< nx\right) = P_{c_0}\left(e^{-zT_n}>e^{ -nz x}\right)
	\ls \e^{nzx} \E_{c_0}\left[\e^{-zT_n}\right]=e^{-n [-zx+\frac 1 2 \log(1+2z)]}.
\end{align*}
The function $g_1(z)=-zx+\frac{1}{2}\log(1+2z)$ attains its maximum at $z_1=\frac{1-x}{2x}>0$ and
\begin{align*}
	g_1(z_{1})=-\frac{1-x}{2} +\frac 1 2\log\left( 1+\frac{1-x}{x}\right) > -\frac{1-x}{2} + \frac{1-x}{2} =0\end{align*}
as $\log(1+y)>\frac{y}{1+y}$ for $y\neq 0$.
Thus, setting $\beta_1=g_1(z_1)>0$, we obtain
\begin{align*}
	\E_{H_0}(\phi_n)=P_{c_0}\left(T_n< nx\right)\ls e^{-n\beta_1}. 
\end{align*}
Similarly, under $H_1$, we get
\begin{align*}
	&T_n\ovunset{=}{\mathcal{D}}{H_1} \bY_n^TB B^T\bY_n,\qquad \text{with } B=S_n^{1/2}D_n(f_1/f_0)^{1/2}S_n^{-1/2}.
\end{align*}
Since 
\begin{align*}
	\mbox{det}(B-\lambda I_n)=\mbox{det}(S_n^{1/2} [D_n^{1/2}(f_1/f_0)-\lambda I_n]S_n^{-1/2})=\mbox{det}(D_n^{1/2}(f_1/f_0)-\lambda I_n)
\end{align*}
the matrix $B$ has the eigenvalues $(c_1(\lambda_i)/c_0(\lambda_i))^{1/2}=(f_{1}(\lambda_i)/f_{0}(\lambda_i))^{1/2}$, $i=1,\ldots,n$.
Since~$B$ is a normal matrix (recall that~$S_n$ is symmetric positive definite),
we find that $BB^T$ has the eigenvalues $c_{1}(\lambda_i)/c_{0}(\lambda_i)$, $i=1,\ldots,n$. 
Consequently,
\begin{align*}
	&T_n\ovunset{=}{\mathcal{D}}{H_1} \bY_n^TB B^T\bY_n\ovunset{=}{\mathcal{D}}{H_1}\bY_n^T D_n(f_1/f_0) \bY_n\ls \bY_n^T\bY_n\,\max_{1\ls i\ls n}\frac{c_{1}(\lambda_i)}{c_{0}(\lambda_i)}  <\bY_n^T\bY_n\,\max_{1\ls i\ls n}\frac{c
_{0}(\lambda_i)-\epsilon}{c_{0}(\lambda_i)}\\
&\ls  \left(1-\frac{\epsilon}{a}\right)\, \bY_n^T\bY_n.
\end{align*}
Consequently,
\begin{align*}
	&\E_{H_1}(1-\phi_n)=P_{c_1}\left(T_n \gs nx\right)\ls P_{c_1}\left( \bY_n^T\bY_n\gs \frac{nx}{1-\eps/a} \right).
\end{align*}
Analogously to the proof under the null hypothesis, we get for  any $z<1/2$ 
\begin{align*}
	&P_{c_1}\left(\bY_n^T\bY_n\gs \frac{n x}{1-\eps/a}\right)\ls e^{-n\left( \frac{zx}{1-\eps/a}+\frac 1 2 \log(1-2z) \right)}.
\end{align*}
Now $g_2(z)=zy+\frac 1 2 \log(1-2z)$, $y=x/(1-\eps/a)$, attains its maximum at $z_2=\frac{y-1}{2y}<\frac 1 2$ with 
\begin{align*}
	g_2(z_2)=\frac{y-1}{2}-\frac 1 2 \log y>0
\end{align*}
as $\log y<y-1$ for $y\neq 1$, i.e.\ $x\neq 1-\eps/a$.
Setting $\beta_2=g_2(z_2)>0$ yields $\E_{H_1}(1-\phi_n)\ls e^{-\beta_2 n}$,
completing the proof.
\end{proof}

\begin{proof}[of Theorem~\ref{theorem:consistency-corrected}]
  We follow \citet{Choudhuri04}, proof of Theorem 1, and show (C1) and
  (C2) of Theorem~\ref{theorem:consistency-general} above.  Let
  $p_n(\bZ_n|c)$ denote the corrected likelihood and define
  $a=||c_{0,\eta}||_{\infty}=\sup_{\lambda\in
    [0,\pi]}|c_{0,\eta}(\lambda)|$ and $b=\inf_{\lambda\in[0,\pi]}
  |c_{0,\eta}(\lambda)|$.  An analogous argument as in Appendix
  B.1. of \citet{Choudhuri04} shows that for all $\epsilon>0$ the set
  $\{ c: ||c-c_{0,\eta}||_\infty<\epsilon\}$ has positive prior
  probability under the Bernstein polynomial prior on $\Theta$.  We
  need to show (C1)(a) and (b).  To this end, let $c\in B$ where
\[B=\{c: ||c-c_{0,\eta}||_\infty <b/2\}.\]
For $c\in B$ we have $c>b/2$.
To prove (C1), first note that
\begin{align*}
		\log \frac{p_n(\bZ_n|c_{0,\eta})}{p_n(\bZ_n|c)}&=\frac 1 2 \sum_{i=1}^n\log \frac{f(\lambda_i)}{f_0(\lambda_i)}-\frac 1 2 \bZ_n^TF_n^TD_n^{-1/2}(f_0)S_n^{-1}D_n^{-1/2}(f_0)F_n\bZ_n\\
		&\quad+\frac 1 2 \bZ_n^TF_n^TD_n^{-1/2}(f)S_n^{-1}D_n^{-1/2}(f)F_n\bZ_n
\end{align*}
where
$S_n=D_n^{-1/2}(f_{\text{param}})F_n\Sigma_{n,\text{param}}F_n^TD_n
(f_{\text{param}})^{-1/2}$.  For $\bY_n\deq N(0,\Sigma)$ it holds $\E
\bY_n^TA\bY_n=\mbox{tr}(A\Sigma)$ as well as $\var(\bY_n^T
A\bY_n)=2\mbox{tr}(A\Sigma A\Sigma)$. Furthermore, by analogous
argument as in the proof of Lemma~\ref{lemma_b3} we get that the
eigenvalues of
\[
  F_n^TD_n(f)^{-1/2}S_n^{-1}D_n(f_0/f)^{1/2}S_nD_n(f_0)^{1/2}F_n
\]
are given by
$f_0(\lambda_i)/f(\lambda_i)=c_{0,\eta}(\lambda_i)/c(\lambda_i)$,
$i=1,\ldots,n$, hence
\begin{align*}
	\mbox{tr}(F_n^TD_n(f)^{-1/2}S_n^{-1}D_n(f_0/f)^{1/2}S_nD_n(f_0)^{1/2}F_n)=\sum_{i=1}^n\frac{c_{0,\eta}(\lambda_i)}{c(\lambda_i)}.
\end{align*}
As a consequence, we get 
\begin{align*}
	0&\ls \frac 1 n K_n(c_0,c)=\frac 1 {2n} \sum_{i=1}^n\log \frac{f(\lambda_i)}{f_0(\lambda_i)}-\frac 1 2 +\frac 1 {2n} \mbox{tr}(F_n^TD_n(f)^{-1/2}S_n^{-1}D_n(f_0/f)^{1/2}S_nD_n(f_0)^{1/2}F_n)\\
	&=\frac 1 {2 n} \sum_{i=1}^n\log\left( \frac{c(\lambda_i)-c_{0,\eta}(\lambda_i)}{c_{0,\eta}(\lambda_i)}+1 \right)+\frac 1{ 2n} \sum_{i=1}^n\left(\frac{c_{0,\eta}(\lambda_i)-c(\lambda_i)}{c(\lambda_i)}+1\right)-\frac 1 2\\
	&=O(1)\;  \|c-c_{0,\eta}\|_{\infty}.
\end{align*}
Similar arguments yield
\begin{align*}
&0\ls \frac 1 {n^2} V_n(c_{0,\eta},c)=\frac 1 {2n^2} \mbox{tr}\left[\left( I_n-F_n^TD_n(f)^{-1/2}S_n^{-1}D_n^{1/2}(f_0/f)S_nD_n(f_0)^{1/2}F_n) \right)^2\right]\\
&=\frac{1}{2n^2}\sum_{i=1}^n\left( \frac{c_{0,\eta}(\lambda_i)-c(\lambda_i)}{c(\lambda_i)} \right)^2   =O(1/n)\;\|c-c_{0,\eta}\|_{\infty}^2
\to 0.
\end{align*}
The proof can now be concluded as in \citet{Choudhuri04} by replacing
Lemma B.3 by Lemma~\ref{lemma_b3} above.
\end{proof}

\section{Appendix: Bayesian autoregressive sampler}\label{sec_appendixSampling}
For fixed order~$p \geq 0$, the autoregressive
model~$Z_t=\sum_{l=1}^pa_lZ_{t-l}+e_t$ with i.i.d. N($0,\sigma^2$)
random variables is parametrized by the the innovation
variance~$\sigma^2$ and the partial autocorrelation structure~$\mathbf
\brho = (\rho_1,\ldots,\rho_p)$, where $\rho_k$ is the conditional
correlation between~$Z_t$ and~$Z_{t+k}$
given~$Z_{t+1},\ldots,Z_{t+k-1}$.  Note that~$\mathbf \brho \in
(-1,1)^p$ and that there is a one-to-one relation between~$\mathbf
a=(a_1,\ldots,a_p)$ and~$\brho$.  To elaborate, we follow
\citet{barndorff1973parametrization} and introduce the auxiliary
variables~$\phi_{k,l}$ as solutions of the following Yule-Walker-type
equation with the autocorrelations $r_k = \E[Z_tZ_{t+k}]/\E Z_1^2$:
\[
  \begin{pmatrix}
    1 & r_1 & \ldots & r_{k-1} \\
    r_1 & 1 & \ldots & r_{k-2} \\
    \vdots & \vdots & \ddots & \vdots \\
    r_{k-1} & r_{k-2} & \ldots & 1
  \end{pmatrix}
  \begin{pmatrix}
    \phi_{k,1} \\
    \phi_{k,2} \\
    \vdots \\
    \phi_{k,k}
  \end{pmatrix}
  =
  \begin{pmatrix}
    r_1\\
    r_2\\
    \vdots\\
    r_k
  \end{pmatrix}.
\]
As shown in~\citet{barndorff1973parametrization}, the well-known relationships
\begin{align*}
  \phi_{p,l} &= a_l,    \quad l=1,\ldots,p, \\
  \phi_{k,k} &= \rho_k, \quad k=1,\ldots,p
\end{align*}
readily imply~$a_p = \rho_p$ and (recall that~$\rho_k = 0$ for~$k > p$) the recursive formula
\[
  a_l = \rho_l - \rho_{l+1} \phi_{l,1} - \rho_{l+2}\phi_{l+1,2} - \ldots - \rho_p \phi_{p-1,p-l}, 
  \quad l=1,\ldots,p-1.
\]
We specify the following prior assumptions for the model parameters:
\begin{align*}
  \rho_l &\sim \text{uniform}((-1,1)), \quad 1 \leq l \leq p, \\
  \sigma^2 & \sim \text{inverse-gamma}(\alpha_\sigma, \beta_\sigma),
\end{align*}
all a priori independent.
We use~$\alpha_\sigma = \beta_\sigma = 0.001$.
Furthermore, we employ the Gaussian likelihood
\begin{align*}
  p(\mathbf Z_n | \mathbf \rho, \sigma^2)
  &= \det \Sigma_{p,\mathbf a, \sigma^2}^{-1/2} 
  \exp\left\{ -\frac{1}{2} \Z_p^T \Sigma_{p,\mathbf a,\sigma^2}^{-1} \Z_p \right\}
 \prod_{i=p+1}^n  \phi_\sigma\left(  Z_i - \sum_{l=1}^p a_l Z_{i-l}  \right)
\end{align*}
with the N($0,\sigma^2$) density~$\phi_\sigma(\cdot)$ and the $p
\times p$ autocovariance matrix $\Sigma_{p,\mathbf a(\mathbf \rho)}$
of the autoregressive model.  To draw samples from the corresponding
posterior distribution, we use a Gibbs sampler, where the conjugate
full conditional for~$\sigma^2$ is readily sampled.
For~$j=1,\ldots,p$ the full conditional for~$\rho_j$ is sampled with
the Metropolis algorithm using a normal random walk proposal density
with proposal variance~$\sigma_l^2$.  As for the parameters of the
working model within the corrected parametric approach (see
Section~\ref{sec_priorParametric}), the proposal variances are
adjusted adaptively during the burn-in period, aiming for a respective
acceptance rate of~0.44.

\section{Appendix: Further simulation results}\label{sec_appendixSims}
Table~2 
depicts further results of the AR, NP and NPC procedures
for normal ARMA data.

\begin{table}
\caption{Average Integrated Absolute Error (aIAE), Uniform Credible Interval coverage (cUCI) and
average posterior model confidence~$\hat \eta$ for
different realizations of model~\eqref{eq_model_ar1}.}\label{tab_iae2}
\centering
\begin{adjustbox}{max width=\linewidth}
\fbox{
\begin{tabular}{l|rrrlrrrlrrr}
&          \multicolumn{3}{c}{White Noise: $a=b=0$}         & & \multicolumn{3}{c}{MA($2$): $b_1=0.75$, $b_2=-0.5$}        & & \multicolumn{3}{c}{AR($2$): $a_1=0.75$, $a_2=-0.5$} \\
\cline{2-4}\cline{6-8}\cline{10-12}
Methods&   $n=64$  &  $n=128$  &  $n=256$    & & $n=64$  &  $n=128$  &  $n=256$   & & $n=64$  &  $n=128$ & $n=256$  \\
\hline 
aIAE &&&&&&&&&&&\\
AR&       0.095  &  0.074   &  0.053     & & 0.343  &  	0.278   &  0.221  & & 0.270  &  0.201  & 0.150  \\
NP&       0.102  &  0.066   &  0.042     & & 0.233  &  	0.177   &  0.128  & & 0.314  &  0.257  & 0.207  \\
NPC&      0.104  &  0.070   &  0.047     & & 0.236  &  	0.180   &  0.137  & & 0.276  &  0.211  & 0.163  \\
\hline
cUCI &&&&&&&&&&& \\
AR&       1.000  &  0.999   &  0.999     & & 0.771  &  	0.757   &  0.786  & & 0.996  &  1.000  & 1.000  \\
NP&       1.000  &  1.000   &  1.000     & & 0.999  &  	0.999   &  0.992  & & 0.998  &  0.992  & 0.982  \\
NPC&      1.000  &  1.000   &  1.000     & & 1.000  &  	1.000   &  1.000  & & 0.999  &  1.000  & 1.000  \\
\hline
$\hat\eta$&0.368 &  0.325   &  0.293     & & 0.293  &   0.199   &  0.123  & & 0.461  &  0.529  & 0.614
\end{tabular}}
\end{adjustbox}
\end{table}

For white noise, all procedures yield good results,
whereas NP is superior due to the implicitly well-specified
white noise working model (recall that the order~$p$ within
AR and NPC are estimated with the DIC).
The results for MA(2) data are qualitatively similar to the results
for MA(1) data in Table~1. 
The spectral peaks of the AR(2) model are not as strong as the
spectral peak of the AR(1) model considered in Table~1. 
Accordingly, it can be seen that the NP results are better in this case.
The NPC benefits again from the well-specified parametric model,
yielding results that are comparable to the AR procedure.

\end{document}